\documentclass[twocolumn,amsmath,amssymb]{revtex4}
\usepackage{graphicx}% Include figure files
\usepackage{dcolumn}% Align table columns on decimal point
\usepackage{bm}% bold math
\usepackage{amsthm}
\usepackage{blackboard}
\newtheorem{example}{Example}[section]
\newtheorem{remark}{Remark}[section]
\newtheorem{theorem}{Theorem}[section]

\newtheorem{lemma}{Lemma}[section]

\def\ket#1{| #1 \rangle}
\def\bra#1{\langle #1 |}
\def\ip#1#2{\langle #1 | #2 \rangle}

\def\norm#1{\| #1 \|}

\def\diag{\operatorname{diag}}
\def\dim{\operatorname{dim}}
\def\rank{\operatorname{rank}}
\def\Span{\operatorname{span}}

\def\Tr{\operatorname{Tr}}
\def\Ad{\operatorname{Ad}}

\def\B{\mathcal{B}}
\def\C{\mathcal{C}}
\def\H{\mathcal{H}}
\def\M{\mathcal{M}}
\def\O{\mathcal{O}}
\def\T{\mathcal{T}}

\def\su{\mathfrak{su}}
\def\SU{\mathfrak{SU}}
\def\UU{\mathfrak{U}}

\def\xs{\vec{x}\cdot\vec{\sigma}}

\begin{document}
\title{Analysis of Lyapunov Method for Control of Quantum States}
\author{Xiaoting Wang}\email{x.wang@damtp.cam.ac.uk}
\affiliation{Department of Applied Maths and Theoretical Physics,
             University of Cambridge, Wilberforce Road, Cambridge, CB3 0WA, UK}
\author{Sonia Schirmer}\email{sgs29@cam.ac.uk}
\affiliation{Department of Applied Maths and Theoretical Physics,
             University of Cambridge, Wilberforce Road, Cambridge, CB3
	     0WA, UK}
\affiliation{Department of Maths and Statistics, 
             University of Kuopio, PO Box 1627, 70211 Kuopio, Finland}
\date{\today}

\begin{abstract}
We present a detailed analysis of the convergence properties of
Lyapunov control for finite-dimensional quantum systems based on the
application of the LaSalle invariance principle and stability
analysis from dynamical systems and control theory.  Under an ideal
choice of the Hamiltonian, convergence results are derived, with a
further discussion of the effectiveness of the method when the ideal
condition of the Hamiltonian is relaxed.
\end{abstract}
\maketitle

\section{Introduction}
\label{sec:intro} 

Control theory has developed into a very broad and interdisciplinary
subject. One of its major concerns is how to design the dynamics of a
given system to steer it to a desired target state, and how to stabilize
the system in a desired state.  Assuming that the evolution of the
controlled system is described by a differential equation, many control
methods have been proposed, including optimal control~\cite{Lewis,Kirk},
geometric control~\cite{Jurdjevic97} and feedback
control~\cite{Franklin}.

Quantum control theory is about the application of classical and modern
control theory to quantum systems.  The effective combination of control
theory and quantum mechanics is not trivial for several reasons.  For
classical control, feedback is a key factor in the control design, and
there has been a strong emphasis on robust control of linear control
systems.  Quantum control systems, on the other hand, cannot usually be
modelled as linear control systems, except when both the system and the
controller are quantum systems and their interaction is fully coherent
or quantum-mechanical~\cite{IEEETAC48p2107}.  This is not the case for
most applications, where we usually desire to control the dynamics of a
quantum system through the interaction with fields produced by what are
effectively classical actuators, whether these be control electrodes or
laser pulse shaping equipment.  Moreover, feedback control for quantum
systems is a nontrivial problem as feedback requires measurements, and
any observation of a quantum system generally disturbs its state, and
often results in a loss of quantum coherence that can reduce the system 
to mostly classical behavior.  Finally, even if measurement backaction
can be mitigated, quantum phenomena often take place on sub-nanosecond
(in many case femto- or attosecond) timescales and thus require ultrafast 
control, making real-time feedback unrealistic at present.

This is not to say that measurement-based quantum feedback control is
unrealistic.  There are various interesting applications, e.g., in the
area of laser cooling of atomic motion~\cite{PRL92n223004}, or for 
deterministic quantum state reduction~\cite{PRL96n010504} and 
stabilization of quantum states, to mention only a few, and progress in
technology will undoubtedly lead to new applications.  Nonetheless,
there are many applications of open-loop Hamiltonian engineering in
diverse areas from quantum chemistry to quantum information processing.
Even in the area of open-loop control many control design strategies,
both geometry~\cite{jurdjevic1,lowenthal,d'alessandro,schirmer} and
optimization-based~\cite{Shi1988,Maday2003,schirmer1}, utilize some form
of model-based feedback.  A particular example is Lyapunov control,
where a Lyapunov function is defined and feedback from a model is used
to generate controls to minimize its value.  Although there have been
several papers discussing the application of Lyapunov control to quantum
systems, the question of when, i.e., for which systems and objectives,
the method is effective and when it is not, has not been answered
satisfactorily.

Several early papers on Lyapunov control for quantum systems such
as~\cite{Vettori,Ferrante,Grivopoulos} considered only control of
pure-state systems, and target states that are eigenstates of the free
Hamiltonian $H_0$, and therefore fixed points of the dynamical system.
For target states that are not eigenstates of $H_0$, i.e., evolve with
time, the control problem can be reformulated either in terms of
asymptotic convergence of the system's actual trajectory to that of the
time-dependent target state, or as convergence to the orbit of the
target state (or more precisely its closure).  Such cases have been
discussed in several papers
~\cite{Mirrahimi2004a,Mirrahimi2004b,Mirrahimi2005,altafini1,altafini2}
but except for~\cite{altafini1,altafini2}, the problem was formulated
using the Schrodinger equation and state vectors that can only represent
a pure state.  To give a complete discussion of Lyapunov control, it is
desirable to utilize the density operator description as it is suitable
for both mixed-state and pure-state systems, and can be generalized to
open quantum systems subject to environmental decoherence or
measurements, including feedback control.  In~\cite{altafini1,altafini2}
Lyapunov control for mixed-state quantum systems was considered but the
notion of orbit convergence used is rather weak compared to trajectory
convergence, the LaSalle invariant set was only shown to contain certain 
critical points but not fully characterized, and a stability analysis of 
the critical points was missing, in addition to other issues such as the 
assumption of periodicity of orbits, etc.  Furthermore, while an attempt 
was made to establish sufficient conditions to guarantee convergence to 
a target orbit, the effectiveness of the method for realistic system was
not considered.

In this paper we address these issues.  We consider the problem of
steering a quantum system to a target state using Lyapunov feedback as a
trajectory tracking problem for a bilinear Hamiltonian control system
defined on a complex manifold, where the trajectory of the target state
is generally non-periodic, and analyze the effectiveness of the Lyapunov
method as a function of the form of the Hamiltonian and the initial
value of the target state.  In Sec.~\ref{sec:basics} the control problem
and the Lyapunov function are defined, and some basic issues such as
different notions of convergence and reachability of target states are
briefly discussed. In Sec.~\ref{sec:LaSalle} the controlled quantum
dynamics is formulated as an autonomous dynamical system defined on an
extended state space, and LaSalle's invariance principle~\cite{lasalle}
is applied to obtain a characterization of the LaSalle invariant set.
This characterization shows that even for ideal systems satisfying the
strongest possible conditions on the Hamiltonian, the invariant set is
generally large, and the invariance principle alone is therefore not
sufficient to conclude asymptotic stability of the target state.  Noting
that the invariant set must contain the critical points of the Lyapunov
function we characterize the former in Sec.~\ref{sec:critical}.  In 
Sec.~\ref{sec:conv_ideal} we give a detailed analysis of the convergence 
behaviour of the Lyapunov method for finite-dimensional quantum systems 
under an ideal control Hamiltonian based on the characterization of
the LaSalle invariant set and our stability analysis.  The discussion is 
divided into three parts, control of pseudo-pure states, generic mixed 
states, and other mixed states.  The result is for this ideal choice of 
Hamiltonian Lyapunov control is effective for most (but not all) target 
states.  Finally, in Sec.~\ref{sec:conv_real} we relax the unrealistic 
requirements on the Hamiltonian imposed in Sec.~\ref{sec:conv_ideal}, 
and show that this leads to a much larger LaSalle invariant set, and 
significantly diminished effectiveness of Lyapunov control.

\section{State and trajectory tracking problem for quantum systems}
\label{sec:basics}

\subsection{Quantum states and evolution}

According to the basic principles of quantum mechanics the state of an
$n$-level quantum system can be represented by an $n\times n$ positive
hermitian operator with unit trace, called a density operator $\rho$,
and its evolution is determined by the Liouville von-Neumann
equation~\cite{von-Neumann}
\begin{eqnarray*}
 \dot \rho(t) = -i\hbar [ H, \rho(t) ],
\end{eqnarray*}
where $H$ is the system Hamiltonian, denoted by an $n\times n$ Hermitian
operator.  If we are considering a sub-system that is not closed, i.e.,
interacts with an external environment, additional terms are required to
account for dissipative effects, although in principle, we can always
consider the Hamiltonian dynamics on an enlarged Hilbert space, and we
shall restrict our discussion here to Hamiltonian systems.  We shall say
a density operator $\rho$ represents a pure state if it is a rank-one
projector, and a mixed state otherwise.  We further define the special
class of pseudo-pure states, i.e., density operators with two
eigenvalues, one of which occurs with multiplicity $1$, the other with
multiplicity $n-1$, and generic mixed states, i.e., density operators
with $n$ distinct eigenvalues.

\subsection{Control Problem}
\label{subsec:problem}

In the following we consider the bilinear Hamiltonian control system
\begin{equation}
  \dot \rho(t) =-i [ H_0+f(t)H_1, \rho(t) ],
\end{equation}
where $f(t)$ is an admissible real-valued control field and $H_0$ and
$H_1$ are a free evolution and control interaction Hamiltonian,
respectively, both of which will be assumed to be time-independent.
We have chosen units such that the Planck constant $\hbar=1$ and can
be omitted for convenience.

The general control problem is to design a certain control function
$f(t)$ such that the system state $\rho(t)$ with $\rho(0)=\rho_0$ will
converge to the target state $\rho_d$.  Since the evolution of a
Hamiltonian system is unitary, the spectrum of $\rho(t)$ is therefore
time-invariant, or equivalently
\begin{equation}
  \Tr[\rho^n(t)]=\Tr[\rho_0^n], \quad \forall n\in \NN.
\end{equation}
Hence, for the target state $\rho_d$ to be reachable, $\rho_0$ and
$\rho_d$ must have the same spectrum, or entropy in physical terms.  If
$\rho_0$ and $\rho_d$ do \emph{not} have the same spectrum, we can still
attempt to minimize the distance $\norm{\rho(t)-\rho_d(t)}$, but it will
always be non-zero if we are restricted to Hamiltonian engineering.  For
the following analysis we shall assume that the initial and the target
state of the system have the same spectrum.  If this is the case and the
system is density-matrix controllable, or pure-state controllable if the
initial state of the system is pure or pseudo-pure, then we can conclude
that the target state is reachable, although a particular target state
may clearly be reachable even if the system is not
controllable~\cite{JPA35p4125}.

Assuming that $\rho_0$ and $\rho_d$ have the same spectrum, the quantum
control problem can be characterized by the spectrum of the target state.
If $\rho_d$ is pure, the problem is called a pure-state control problem.
Analogously, we can define the pseuo-pure-state control and generic-state
control.  Pure-state control problems are often represented in terms of
Hilbert space vectors or wavefunctions $\ket{\psi}$ evolving according to
the Schr\"odinger equation
\begin{eqnarray}
   \frac{d}{dt}\ket{\psi(t)} = -i ( H_0+f(t)H_1) \ket{\psi(t)}.
\end{eqnarray}
For pure states this wavefunction descritpion is equivalent to the density
operator description since any rank-one projector $\rho$ can be written
as $\rho=\ket{\psi}\bra{\psi}$ for some Hilbert space vector $\ket{\psi}$,
but it does not generalize to mixed states, and we shall not use this
formalism here.

Since the free Hamiltonian $H_0$ can usually not be turned off, it is
natural to consider non-stationary target states $\rho_d$ evolving
according to
\begin{equation}
\label{eqn:3}
\dot \rho_d(t) =-i [ H_0, \rho_d(t) ].
\end{equation}
It is easy to see that $\rho_d$ is stationary if and only if it commutes
with $H_0$, $[H_0,\rho_d(0)]=0$.  Thus the problem of quantum state control
for most target states is more akin to a trajectory tracking problem, where
the objective generally is to find a control $f(t)$ such that the trajectory
$\rho(t)$ of the initial state $\rho_0$ under the controlled evolution
asymptotically converges to a target trajectory $\rho_d(t)$.

\subsection{Trajectory vs Orbit Tracking}
\label{subsec:tracking}

It has been argued that the problem of quantum state control should
instead be viewed as an orbit tracking problem~\cite{altafini1,altafini2},
i.e., the problem of steering the trajectory $\rho(t)$ towards the orbit
of the target state $\rho_d$.  However, one problem with this approach
is that the notion of orbit tracking is relatively weak as the orbit of
a quantum state, or more precisely its closure, under free evolution can
be rather large, and there are generally infinitely many distinct quantum
states whose orbits under free evolution coincide.  For example, even for
two-level system evolving under the free Hamiltonian $H_0=\diag(0,\omega)$
the trajectories of the pure states
$\ket{\Psi_\pm} = \frac{1}{\sqrt{2}}(\ket{0}\pm \ket{1})$
are orthogonal, and thus perfectly distinguishable, for all times $t$,
$\ket{\Psi_\pm(t)} = \frac{1}{\sqrt{2}}(\ket{0}\pm e^{i\omega t}\ket{1})$,
but their orbits are the same,
$\O(\Psi_+)=\{\Psi(t): \Psi(0)=\Psi_+, t\ge 0\} = \O(\Psi_-)$.

For the two-level example above, the orbits are always periodic and thus
closed, and we can at least say that if the quantums state $\rho(t)$
converges to the \emph{periodic orbit} $\O(\rho_d)$ of $\rho_d$, then
for every state $\rho_a\in \O(\rho_d)$ there exists a sequence of times
$\{t_k\}$ such that $\norm{\rho(t_k)-\rho_a}\to 0$ as $k\to\infty$, but
this notion of convergence is much weaker than the notion of trajectory
convergence, which requires $\norm{\rho(t)-\rho_d(t)}\to 0$ as
$t\to\infty$, and we shall see that there are cases where it is possible
to track the orbit but \emph{not} a particular trajectory.
The notion of orbit tracking is even more problematic for non-periodic
orbits, which comprise the vast majority of orbits for systems of Hilbert
dimension $n>2$, except for the measure-zero set of Hamiltonians $H_0$
with commensurate energy levels, i.e., with transition frequencies that
are rational multiples of each other.  Of course, we can still ask the
question whether the state of the system converges to the \emph{closure}
of the orbit of a target state, but the dimension of this orbit set is
generally very large.  For instance, the state manifold of pure states
for an $n$-dimensional system has (real) dimension $2n-2$, while the
closure of the orbit of any state under a generic Hamiltonian $H_0$ has
dimension $n-1$.

For these reasons, we shall concentrate on quantum state control in
the sense of trajectory tracking as this is the strongest notion of
convergence and well-defined for arbitrary trajectories.

\subsection{Control Design based on Lyapunov Function}
\label{subsec:Vdef}

A natural design of $f(t)$ is inspired from the conception of Lyapunov
function, which is a very important tool in stability analysis for
dynamical systems. For an autonomous dynamical system $\dot x=f(x)$, a
differentiable scalar function $V(x)$, defined on the phase space
$\Omega=\{x\}$, is called a Lyapunov function, if:
\begin{enumerate}
\item[(i)] $V(x)$ is continuous and its partial derivatives are also
continuous on $\Omega$;
\item[(ii)] $V(x)$ is positive definite, i.e., $V(x)\ge 0$ with equality
only at $x=x_0$;
\item[(iii)] for any dynamical flow $\phi_t(x)$,
$\dot V(\phi_t(x))=\dot V(x(t)) \le 0$.
\end{enumerate}
With the conditions above, it can be shown that $x=x_0$ is Lyapunov
stable; if equality in (iii) holds only for $x=x_0$, we can further
conclude that $x=x_0$ is asymptotically stable.  However, in general, we
can only guarantee $\dot V \le 0$, and in this case, we can only use a
weaker result known as the LaSalle invariance principle~\cite{lasalle},
which claims that any bounded solution will converge to an invariant
set, called the LaSalle invariant set. 

Let $\M$ to be the set of density operators isospectral with $\rho_d(0)$ 
and consider the joint dynamics for $(\rho(t),\rho_d(t))$ on $\M\times\M$:
\begin{subequations}
\label{eqn:auto0}
\begin{align}
\dot \rho(t) &=-i [ H_0+f(t)H_1, \rho(t) ],\\
\dot \rho_d(t) &=-i [ H_0, \rho_d(t) ].
\end{align}
\end{subequations}
The Hilbert-Schmidt norm $\norm{A}=\sqrt{\Tr(A^\dagger A)}$ induces a 
natural distance function on $\M\times\M$, which provides a natural
candidate for a Lyapunov function
\begin{equation}
\label{eqn:4}
  V(\rho,\rho_d) = \frac{1}{2}\norm{\rho-\rho_d}^2
                 = \frac{1}{2}\Tr[(\rho-\rho_d)^2].
\end{equation}
If $\rho$ and $\rho_d$ are isospectral, this definition is equivalent to
\begin{equation}
 \label{eq:4a}
 V(\rho,\rho_d) = \Tr[\rho_d^2(t)]- \Tr[\rho(t)\rho_d(t)],
\end{equation}
the Lyapunov function used in~\cite{altafini1,altafini2}.  If
$\rho_d=\ket{\psi_d}\bra{\psi_d}$ and $\rho=\ket{\psi}\bra{\psi}$ we
have furthermore
\begin{equation}
 \label{eq:4b}
 V(\psi,\psi_d) = 1- |\ip{\psi_d(t)}{\psi(t)}|^2,
\end{equation}
a Lyapunov function often used for pure-state control.

To see that Eq.~(\ref{eq:4a}) defines indeed a Lyapunov function, note
that $V\ge 0$ with equality only if $\rho=\rho_d$, and
\begin{align*}
{\dot V} &= -\Tr(\dot\rho_d \rho)-\Tr(\rho_d \dot \rho)\\
         &= -\Tr([-iH_0,\rho_d]\rho)-\Tr(\rho_d[-iH_0,\rho])\\
         &\qquad -f(t)\Tr(\rho_d[-iH_1,\rho])\\
         &= -f(t)\Tr(\rho_d[-iH_1,\rho]),
\end{align*}
where we have used
\[
  \Tr([-iH_0,\rho_d]\rho) = -\Tr(\rho_d[-iH_0,\rho])
\]
and $\frac{d}{dt}\Tr(\rho^2_d)=0$.  If we choose the control field as
\begin{equation}
 \label{eqn:5}
  f(\rho,\rho_d) = \kappa \Tr(\rho_d[-iH_1,\rho]), \quad \kappa>0,
\end{equation}
then $\dot V(\rho(t),\rho_d(t)) \le 0$.  Without loss of generality,
we set $\kappa=1$ in the following.  

Hence, the evolution of the system $(\rho,\rho_d)$ with Lyapunov
feedback is described by the following nonlinear autonomous dynamical
system on $\M\times\M$:
\begin{subequations}
\label{eqn:auto}
\begin{align}
 \dot{\rho(t)}   &= -i [ H_0+f(\rho,\rho_d)H_1, \rho(t) ],\\
 \dot{\rho_d(t)} &= -i [ H_0, \rho_d(t) ],\\
 f(\rho,\rho_d)  &= \Tr([-iH_1,\rho]\rho_d).
\end{align}
\end{subequations}
The manifold $\M$ here is a homogeneous space known as a flag manifold,
whose dimension and topology depend on the spectrum, or more precisely,
the number of distinct eigenvalues, of the density operators $\rho$, 
$\rho_d$, under consideration.  For pure or pseudo-pure initial states 
$\rho_0$, for example, $\M$ is homeomorphic to the complex projective 
space $\CC P^{n-1}$, while for a generic mixed state, we obtain the 
$n^2-n$ dimensional manifold $U(n)/\oplus_{\ell=1}^n U(1)$.  By simply 
comparing the dimensions, we see that in the special case $n=2$ (and 
only $2$) the generic mixed states and pseudo-pure states have the same 
dimension, and one can easily show that in this case all mixed states 
(except the completely mixed state) are pseudo-pure, a fact that will 
be relevant later.

\section{LaSalle Invariance Principle and LaSalle Invariant Set}
\label{sec:LaSalle}

\subsection{Invariance Principle for Autonomous Systems}
\label{subsec:invar}

For an autonomous dynamical system with $\dot x=f(x)$, we say a set is
\emph{invariant}, if any flow starting at a point in the set will stay
in it for all times.  For any solution $x(t)$ we define the positive
limiting set $\Gamma^+$ to be the set of all limit points of $x(t)$ as
$t\to+\infty$.  First of all, we have the following two lemmas:

\begin{lemma}
For $\dot x=f(x)$ defined on a finite-dimensional manifold, the
positive limiting set $\Gamma^+$ of any bounded solution $x(t)$ is
an non-empty, connected, compact, invariant set.
\end{lemma}
The proof can be found in~\cite{Perko} (Sec.~3.2, Theorem~$2$).

\begin{lemma}
\label{lemma:2}
Any bounded solution $x(t)$ will tend to any set containing its
positive limiting set $\Gamma^+$ as $t \to \infty$.
\end{lemma}

\begin{proof}
Suppose $x(t)$ does not converge to $\Gamma^+$.  Then there exists
some $\epsilon>0$, and a sequence $t_n$ such that $x(t_n)$ is outside
the $\epsilon$-neighborhood of $\Gamma^+$.  But $x(t_n)$ is a bounded
set, so it has a subsequence that converges to a point $x_0$.  By
assumption $x_0\not\in\Gamma^+$, which contradicts the definition of
the positive limiting set.  Hence, $x_0$ must belong to the positive
limiting set.
\end{proof}

From these results we can derive the LaSalle invariance 
principle~\cite{lasalle}:

\begin{theorem}
\label{thm:lasalle:1}
For an autonomous dynamical system, $\dot x=f(x)$, let $V(x)$ be
a Lyapunov function on the phase space $\Omega=\{x\}$, satisfying
$V(x)>0$ for all $x \neq x_0$ and $\dot{V}(x) \le 0$, and let
$\O(x(t))$ be the orbit of $x(t)$ in the phase space.
Then the invariant set $E=\{\O(x(t))|\dot{V}(x(t))=0\}$
contains the positive limiting sets of all bounded solutions,
i.e., any bounded solution converges to $E$ as $t\to+\infty$.
\end{theorem}

\begin{proof}
Since $V(x(t))$ is monotonically decreasing due to $\dot V \le 0$,
$V(x(t))$ has a limit $V_0\ge0$ as $t \to + \infty$ for any bounded
solution $x(t)$. Let $\Gamma^+$ be the positive limiting set of
$x(t)$. By continuity, the value of $V$ on $\Gamma^+$ must be $V_0$.
Since $\Gamma^+$ is an invariant set, we can take the time derivative
of $V$ to conclude $\dot V=0$ on $\Gamma^+$. By Lemma~\ref{lemma:2},
$x(t)$ will converge to $\Gamma^+$, and hence to $E$.
\end{proof}

\begin{remark}
From the proof above, we can see that the theorem holds for both
real and complex dynamical systems.  Broadly speaking, what has been
proved is that bounded solutions with $\dot V(x)\ne 0$ will converge
to the set of solutions with $\dot V(x)=0$.  Therefore, it does not
matter if $V$ has many points $x$ with $V(x)=0$. For example, for
the quantum system~(\ref{eqn:auto}), the Lyapunov function $V$ is
zero on all points $(\rho_d, \rho_d)$.
\end{remark}

The quantum system~(\ref{eqn:auto}) is autonomous and defined on the
phase space $\M \times \M$, where $\M$ is a compact finite dimensional
manifold.  Therefore, any solution $(\rho(t),\rho_d(t))$ is bounded.
Although the Lyapunov function~(\ref{eq:4a}) is not positive definite,
we have $V=0$ if and only if $\rho=\rho_d$, which is sufficient to 
apply the LaSalle invariance principle~\ref{thm:lasalle:1} to obtain:

\begin{theorem}
\label{thm:lasalle:2}
Any system evolution $(\rho(t),\rho_d(t))$ under the Lyapunov
control (\ref{eqn:5}) will converge to the invariant set
$E=\{(\rho_1,\rho_2)\in \mathcal{M} \times
\mathcal{M}|\dot{V}(\rho(t),\rho_d(t))=0,
(\rho(0),\rho_d(0))=(\rho_1,\rho_2)\}$.
\end{theorem}

We note here that except when $\rho_d$ is a stationary state, we must 
consider the dynamical system on the extended phase space $\M\times\M$ 
as $V$ is \emph{not} well-defined on $\M$.  Having established 
convergence to the LaSalle invariant set $E$, the next step is to 
characterize $E$ for the dynamical system~(\ref{eqn:auto}).

\subsection{Characterization of the LaSalle Invariant Set}
\label{subsec:charE}

LaSalle's invariance principle reduces the convergence analysis to
calculating the invariant set $E=\{\dot V(\rho(t),\rho_d(t))=0\}$, which
is equivalent to $f(t)=0$, for any $t$.  Therefore, we have
\begin{equation*}
\begin{aligned}
0 &= f      = \Tr([-iH_1,\rho]\rho_d)\\
0 &= \dot f = \Tr([-iH_1,\rho]\dot \rho_d)+\Tr([-iH_1,\dot \rho]\rho_d)\\
  & \qquad   -\Tr([[-iH_0,-iH_1],\rho]\rho_d)\\
  & \cdots \\
0 &= \frac{d^{\ell}f}{dt^\ell}=(-1)^n \Tr([\Ad^\ell_{-iH_0}(-iH_1),\rho]\rho_d),
\end{aligned}
\end{equation*}
where $\Ad^\ell_{-iH_0}(-iH_1)$ represents $\ell$-fold commutator adjoint
action of $-iH_0$ on $-iH_1$.  Hence,
$\Tr([A,B]C)=-\Tr([C,B]A)=-\Tr([A,C]B)$ gives a necessary condition for
the invariant set $E$:
\begin{equation}
\label{eq:trace-cond1}
 \Tr([\rho,\rho_d]\Ad^m_{-iH_0}(-iH_1))=0, \qquad \forall m\in\NN_0,
\end{equation}
where $\Ad^0_{-iH_0}(-iH_1)=-iH_1$.  Since $H_0$ is Hermitian we can choose
a basis such that $H_0$ is diagonal
\begin{equation*}
H_0=
\begin{pmatrix}
  a_1 & 0   & \ldots & 0\\
  0   & a_2 & \ldots & 0\\
  \vdots   &     & \ddots & \vdots\\
  0   & 0   & \ldots & a_n
\end{pmatrix}
\equiv \diag(a_1,\ldots,a_n)
\end{equation*}
with real eigenvalues $a_k$, which we may assume to be arranged so that
$a_k\ge a_{k+1}$ for all $k$. Let $(b_{k\ell})$ be the matrix
representation of $H_1$ in the eigenbasis of $H_0$, and
$\omega_{k\ell}=a_\ell-a_k$ be the transition frequency between energy
levels $k$ and $\ell$ of the system.

The Lie algebra $\su(n)$ can be decomposed into an abelian part called the
Cartan subalgebra $\C=\Span\{\lambda_k\}_{k=1}^{n-1}$, and an orthogonal
subalgebra $\T$, which is a direct sum of $n(n-1)/2$ root spaces spanned
by pairs of generators  $\{\lambda_{k\ell},\bar{\lambda}_{k\ell}\}$.
For instance, we can choose the generators
\begin{subequations}
\label{eq:lambda}
\begin{align}
 \lambda_k            &=  i(\hat{e}_{kk} - \hat{e}_{k+1,k+1}) \\
 \lambda_{k\ell}      &=  i(\hat{e}_{k\ell}+\hat{e}_{\ell k})\\
 \bar{\lambda}_{k\ell}&=   (\hat{e}_{k\ell}-\hat{e}_{\ell k})
\end{align}
\end{subequations}
for $1\le k<\ell\le n$, where the $(k,\ell)^{\rm th}$ entry of
the elementary matrix $\hat{e}_{mn}$ equals $\delta_{km}\delta_{\ell n}$.
Expanding $-iH_1\in\su(n)$ with respect to these generators
\begin{equation}
 \label{eq:H1}
  -i H_1 = \sum_{k=1}^{n-1}\left[ b_k \lambda_k + \sum_{\ell=k+1}^n
           -\Re(b_{k\ell}) \lambda_{k\ell} + \Im(b_{k\ell}) \bar{\lambda}_{k\ell}\right]
\end{equation}
and noting that we have for $D=\sum_{k=1}^n d_k\hat{e}_{kk}$
\begin{subequations}
\label{eq:lambda_comm}
\begin{align}
  [D,\lambda_k]             &= 0,\\
  [D,\lambda_{k\ell}]       &= +i(d_k-d_\ell) \bar{\lambda}_{k\ell}, \\
  [D,\bar{\lambda}_{k\ell}] &= -i(d_k-d_\ell) \lambda_{k\ell},
\end{align}
\end{subequations}
shows that $B_m=\Ad_{-iH_0}^m(-iH_1)$ is equal to
\begin{subequations}
\label{eq:Bm}
\begin{align}
  B_{2m-1} &= \sum_{k=1}^{n-1}\sum_{\ell=k+1}^n \!\!\!(-1)^m \omega_{k\ell}^{2m-1}
                       [\Re(b_{k\ell}) \bar{\lambda}_{k\ell}
                        +\Im(b_{k\ell}) \lambda_{k\ell}], \\
  B_{2m}   &= \sum_{k=1}^{n-1}\sum_{\ell=k+1}^n (-1)^m \omega_{k\ell}^{2m}
                       [\Re(b_{k\ell}) \lambda_{k\ell}
                        -\Im(b_{k\ell}) \bar{\lambda}_{k\ell}].
\end{align}
\end{subequations}
Let$\B^s=\Span\{B_m\}_{m=1}^s$ and $\B_0^s=\Span\{B_m\}_{m=0}^s$
with $B_0=-iH_1$.  Then Eq~(\ref{eq:trace-cond1}) is equivalent to
$[\rho,\rho_d]$ being orthogonal to the subspace $\B_0^s$ with
respect to the Hilbert-Schmidt norm.

\begin{theorem}
\label{thm:lasalle:3}
The subspace $\B^{n^2-n}$ generated by the Ad-brackets is a subset of
the Cartan subalgebra $\T$ of $\su(n)$ with equality if
\begin{itemize}
\item[(i)] $H_0$ is strongly regular, i.e.,
           $\omega_{k\ell}\neq \omega_{pq}$ unless $(k,\ell)=(p,q)$.
\item[(ii)] $H_1$ is fully connected, i.e., $b_{k\ell}\neq 0$ except
            (possibly) for $k=\ell$.
\end{itemize}
\end{theorem}

\begin{proof}
Since the dimension of $\T$ is $n^2-n$ and $B_m\in\T$ for all $m>0$,
it suffices to show that the elements $B_m$ for $m=1,\ldots,n^2-n$ are
linearly independent.  Moreover, the subspaces spanned by the odd and
even order elements, $\B_{\rm odd}^s=\Span\{B_{2m-1}: 1\le 2m-1\le s\}$
and $\B_{\rm even}^s = \Span\{B_{2m}: 1\le 2m\le s\}$, respectively, are
orthogonal since
\begin{multline*}
 B_{2m-1} B_{2m'}
 = (-1)^{m+m'}\sum_{k,\ell}\sum_{k',\ell'}
   \omega_{k\ell}^{2m-1}
   \omega_{k'\ell'}^{2m'} \times \\
  [\Re(b_{k\ell}) \Re(b_{k'\ell'}) \bar{\lambda}_{k\ell} \lambda_{k'\ell'}
  -\Im(b_{k\ell}) \Im(b_{k'\ell'}) \lambda_{k\ell} \bar{\lambda}_{k'\ell'}\\
  -\Re(b_{k\ell}) \Im(b_{k'\ell'}) \bar{\lambda}_{k\ell}\bar\lambda_{k'\ell'}
  +\Im(b_{k\ell}) \Re(b_{k'\ell'}) \lambda_{k\ell} \lambda_{k'\ell'}],
\end{multline*}
and thus observing the equalities
\begin{subequations}
\begin{gather}
  \Tr(\lambda_{k\ell}\lambda_{k'\ell'})
  = \Tr(\bar\lambda_{k\ell}\bar\lambda_{k'\ell'})
  = -2\delta_{kk'}\delta_{\ell\ell'} \\
   \Tr(\lambda_{k\ell}\bar\lambda_{k'\ell'}) =0
\end{gather}
\end{subequations}
shows that for all $m,m'>0$
\begin{multline*}
 \Tr(B_{2m-1} B_{2m'})
 = (-1)^{m+m'}\sum_{k,\ell}\sum_{k',\ell'}
   \omega_{k\ell}^{2m-1} \omega_{k'\ell'}^{2m'} \times \\
   \Re(b_{k\ell}) \Im(b_{k\ell})[-\bar{\lambda}_{k\ell}^2+\lambda_{k\ell}^2]=0.
\end{multline*}
Thus it suffices to show that the elements of $\B_{\rm odd}^{n^2-n}$
and $\B_{\rm even}^{n^2-n}$ are linearly independent separately.

For the odd terms, suppose there exists a vector
$\vec{c}=(c_1,\ldots,c_s)^T$ of length $s=n(n-1)/2$ such that
$\sum_{m=1}^s c_m B_{2m-1}=0$.  Noting that $\omega_{kk}=0$ and
$(\omega_{\ell k})^2=(-\omega_{k\ell})^2$ this gives $n(n-1)/2$
non-trivial equations
\begin{equation}
\label{eq:lin_indep}
  \omega_{k\ell} [\Re(b_{k\ell}) \bar{\lambda}_{k\ell}
                        +\Im(b_{k\ell}) \lambda_{k\ell}] \;
  \sum_{m=1}^{s} (-\omega_{k\ell}^2)^{m-1} c_m=0,
\end{equation}
for $1\le k<\ell\le n$.  Since $\omega_{k\ell}\neq 0$,
$b_{k\ell}\neq0$ by hypothesis, Eq.~(\ref{eq:lin_indep}) can be
reduced to $\Omega \vec{c}=\vec{0}$, where $\Omega$ is a matrix:
\begin{equation}
  \label{eq:vandermonde}
  \Omega = \begin{pmatrix}
    1 & -\omega_{12}^2 & \omega_{12}^4 & \ldots & (-\omega_{12}^2)^{m-1} \\
    1 & -\omega_{13}^2 & \omega_{13}^4 & \ldots & (-\omega_{13}^2)^{m-1} \\
    \vdots & \vdots & \vdots & \ddots & \vdots \\
    1 & -\omega_{n-1,n}^2 & \omega_{n-1,n}^4 & \ldots & (-\omega_{n-1,n}^2)^{m-1} \\
      \end{pmatrix}.
\end{equation}
Since $\Omega$ is a Vandermonde matrix, condition~(ii) of the
proposition guarantees that Eq.~(\ref{eq:lin_indep}) has only the
trivial solution $\vec{c}=\vec{0}$, thus establishing linear
independence.  For the even terms we obtain a similar system of
equations, which completes the proof.
\end{proof}

If $\B^{n^2-n}=\T$ then any point $(\rho_1,\rho_2)$ in the invariant
set $E$ must satisfy $[\rho_1,\rho_2]=\diag(c_1,\ldots,c_n)$.
Furthermore, $\Tr(-iH_1 [\rho_1,\rho_2])=0$ yields in addition
\begin{equation}
  -i\sum_{k,\ell=1}^n b_{k\ell}c_{\ell k}=-i\sum_{k=1}^n b_{kk}c_{kk}=0.
\end{equation}
However, in many applications the energy level shifts induced by the
field are negligible, and we can assume the diagonal elements of $H_1$ 
to be zero.  With this additional assumption we have $\B_0^s\subset\T$, 
and thus the maximum dimension of $\B_0^s$ is $n^2-n$, and we have the 
following useful result.
\begin{theorem}
\label{thm:lasalle:4} 
Under conditions (i) and (ii) of Theorem~\ref{thm:lasalle:3}
$(\rho_1,\rho_2)$ belongs to the invariant set $E$ if and only if
$[\rho_1,\rho_2]=\diag(c_1,\ldots,c_n)$.
\end{theorem}

\begin{proof}
We have proved the necessary part.  For the sufficient part note
that $\rho_k(t)=e^{-iH_0t}\rho_k e^{iH_0t}$, $k=1,2$,
\[
 [e^{-iH_0t}\rho_1 e^{iH_0t},e^{-iH_0t}\rho_2e^{iH_0t}]
  = e^{-iH_0t}[\rho_1,\rho_2] e^{iH_0t}
\]
and $e^{-iH_0 t}$ diagonal.  Thus if $[\rho_1,\rho_2]=
\diag(c_1,\ldots,c_n)$ then
\(
  e^{-iH_0t}[\rho_1,\rho_2] e^{iH_0t} =
 \diag(c_1,\ldots,c_n) = [\rho_1,\rho_2]
\)
and hence $(\rho_1,\rho_2)\in E$.
\end{proof}

Thus we have fully characterized the invariant set for systems with 
strongly regularly $H_0$ and an interaction Hamiltonian $H_1$ with a 
fully connected transition graph.  The result also shows that even 
under the most stringent assumptions about the system Hamiltonians,
the invariant set is generally much larger than the desired solution.
Therefore, the invariance principle alone is not sufficient to
establish convergence to the target state.  

\section{Critical Points of the Lyapunov Function}
\label{sec:critical}

In this section we show that invariant set $E$ always contains at least
the critical points of the Lyapunov function $V$ and classify the
stability of the critical points.  We start with the case where $\rho_d$
is a fixed stationary state.  In this case the Lyapunov function
$V(\rho,\rho_d)$ is effectively a function $V(\rho)$ on $\M$.  Since
$\rho$ can be written as $\rho=U\rho_d U^\dagger$ for some $U$ in the
special unitary group $\SU(n)$, $V$ also be considered a function on 
$\SU(n)$, $V(U)=V(U\rho_dU^\dagger\rho_d)$.  It is easy to see that the
critical points of $V(\rho)$ correspond to those of $V(U)$, and since
$\Tr[\rho_d(t)]^2=C$ is constant, it is equivalent to find the critical
points of $J(U)$:
\begin{align}
  \label{eq:J}
  J(U) &=C-V(U) \nonumber\\
  &= \Tr(U \rho_d U^\dagger \rho_d), \quad U\in\SU(n).
\end{align}

\begin{lemma}
\label{lemma:crit:1} 
The critical points $U_0$ of $J(U)$ defined by~(\ref{eq:J}) are such
that $[\rho_0,\rho_d]=0$ for $\rho_0=U_0\rho_d U_0^\dagger$.
\end{lemma}

\begin{proof}
Let $\{\sigma_m\}_{m=1}^{n^2-1}$ be an orthonormal basis for the Lie
algebra $\su(n)$, consisting of $n^2-n$ orthonormal off-diagonal
generators such as $\frac{1}{\sqrt{2}}\lambda_{k\ell}$,
$\frac{1}{\sqrt{2}}\bar{\lambda}_{k\ell}$ with
$\lambda_{k\ell},\bar{\lambda}_{k\ell}$ as in Eq.~(\ref{eq:lambda}), and
$n-1$ orthonormal diagonal generators
\begin{equation}
 \sigma_{n^2-n+r} = \frac{i}{\sqrt{r(r+1)}}
                     \left(\sum_{s=1}^{r} \hat{e}_{ss}
                     -r\hat{e}_{r+1,r+1} \right)
\end{equation}
for $r=1,\ldots,n-1$.  Set
$\vec{\sigma}=(\sigma_1,\ldots,\sigma_{n^2-1})$.  Any $U\in\SU(n)$ near
the identity $I$ can be written as $U=e^{\vec x\cdot\vec\sigma}$, where
$\vec{x}\in\RR^n$ is the coordinate of $U$, and any $U$ in the 
neighborhood of $U_0$ can be parameterized as 
$U=e^{\vec{x}\cdot\vec\sigma}U_0$.  Thus Eq.~(\ref{eq:J}) becomes
\begin{equation}
\label{eq:J:local} J = \Tr[ (e^{\xs} U_0) \rho_d (U_0^\dagger e^{-\xs})
  \rho_d].
\end{equation}
At the critical point $U_0$, $\nabla J=0$ implies that for all $m$
\begin{align}
 0= \frac{\partial J}{\partial x_m}
  &= \Tr(\sigma_m U_0\rho_d U_0^\dagger \rho_d
     -U_0\rho_d U_0^\dagger\sigma_m\rho_d)\nonumber\\
  &= \Tr[\sigma_m(U_0\rho_d U_0^\dagger \rho_d
     -\rho_d U_0\rho_d U_0^\dagger )]\nonumber\\
  &= \Tr(\sigma_m[U_0\rho_d U_0^\dagger, \rho_d]).
\end{align}
Thus $[U_0\rho_d U_0^\dagger,\rho_d]\in\su(n)$ is orthogonal to all
basis elements $\sigma_m$, and therefore
$[U_0\rho_d U_0^\dagger,\rho_d]=0$.
\end{proof}

Hence, for a given $\rho_d$, the critical points of $V(\rho)$ are such
that $[\rho,\rho_d]=0$, i.e., $\rho$ and $\rho_d$ are simultaneously
diagonalizable.  Let $\{w_1,\ldots,w_n\}$ be the spectrum of $\rho_d$
with $w_k$ arranged in a non-increasing order.  For any critical point
$\rho_0$ there thus exists a basis such that
\begin{align*}
\rho_d &=\diag(w_1,\ldots,w_n),\\
\rho_0 &=\diag(w_{\tau(1)},\ldots,w_{\tau(n)}),
\end{align*}
for some permutation $\tau$ of the numbers $\{1,\ldots,n\}$, and the
corresponding critical value of $V$ is
\begin{equation}
\label{eqn:V}
  V(\rho_0,\rho_d) = \sum_{k=1}^n w_k(w_k - w_{\tau(k)}).
\end{equation}

More generally, for $V(\rho_1,\rho_2)$ defined on $\M\times \M$, there 
exists $U_1$ and $U_2$ such that
\begin{align*}
\rho_1 = U_1\rho_d U_1^\dagger, \quad \rho_2 = U_2\rho_dU_2^\dagger.
\end{align*}
Since $\Tr(\rho_2^2)$ is constant, the critical points of $V$ are again
the critical points of $J(\rho_1,\rho_2)=\Tr(\rho_1\rho_2)$ and
\begin{align*}
J(U_1,U_2) &=\Tr\Big( U_1\rho_d U_1^\dagger U_2\rho_dU_2^\dagger\Big)\\
           &=\Tr\Big((U_2^\dagger U_1)\rho_d(U_2^\dagger U_1)^\dagger\rho_d\Big)
\end{align*}
together with Lemma~\ref{lemma:crit:1} shows that $J$ attains its
critical value when  
$[(U_2^\dagger U_1)\rho_d(U_2^\dagger U_1)^\dagger, \rho_d]=0$, and thus
\begin{align*}
 0 &= U_2 [(U_2^\dagger U_1)\rho_d(U_2^\dagger U_1)^\dagger, \rho_d]U_2^\dagger\\
   &= [U_1\rho_d U_1^\dagger,U_2 \rho_d U_2^\dagger] = [\rho_1,\rho_2].
\end{align*}
Thus we have the following:

\begin{theorem}
\label{thm:crit:1}
For a given $\rho_d$, the critical points of the Lyapunov function
$V(\rho_1,\rho_2)$ on $\M \times \M$ are
such that $\{[\rho_1,\rho_2]=0\}$. Therefore, the LaSalle invariant
set contains all the critical points of $V(\rho_1,\rho_2)$.
\end{theorem}

Next, we show that for a generic stationary state $\rho_d$,
$J(\rho)=\Tr(\rho\rho_d)$, and thus $V(\rho)=V(\rho,\rho_d)$, is a
Morse function~\cite{Morse} on $\M$, i.e., its critical points are
hyperbolic:
\begin{theorem}
\label{thm:crit:2}
If $\rho_d$ has non-degenerate eigenvalues then $J(\rho)$ is a Morse 
function on $\M$.  Moreover, all but two critical points corresponding 
to the global maximum and minimum of $J$, respectively, are saddle 
points with critical values $J_0$ satisfying 
$J_{\rm min}< J_0 < J_{\rm max}$.
\end{theorem}

\begin{proof}
For non-degenerate $\rho_d$, we choose a basis such that
$\rho_d=\diag(w_1,\ldots,w_n)$ with $w_k$ arranged in decreasing order.
Then there are $n!$ critical points satisfying
$\rho_0=\diag(w_{\tau(1)},\ldots,w_{\tau(n)})$, for some permutation
$\tau$, corresponding to the critical value $J(\rho_0) = \sum_{k=1}^n
w_k w_{\tau(k)}$.  Again, we consider
$J=\Tr(\rho\rho_d)=\Tr(U\rho_dU^\dagger\rho_d)$ as a function on
$\SU(n)$.  Let $U_0$ correspond to the critical point $\rho_0$.  As in
the proof of Theorem \ref{thm:crit:1}, any $U$ in the neighborhood of
$U_0$ can again be parameterized as $U=e^{\xs}U_0$. Substituting this
into $J$, we obtain:
\begin{align*}
J &=\Tr[e^{\xs}U_0\rho_d U_0^\dagger e^{-\xs}\rho_d]\\
  &=\Tr[(I+\xs+\frac{1}{2}(\xs)^2) \times U_0\rho_d U_0^\dagger \times \\
  &\qquad (I-\xs+\frac{1}{2}(\xs)^2)\rho_d]+\Theta(|\vec x|^3)\\
  &=\Tr[U_0\rho_d U_0^\dagger\rho_d]+ \frac{1}{2}\Tr[(\xs)^2 U_0\rho_d U_0^\dagger\rho_d]\\
  &\quad +\frac{1}{2}\Tr[U_0\rho_d U_0^\dagger(\xs)^2\rho_d]\\
  &\quad -\Tr[(\xs)U_0\rho_d U_0^\dagger(\xs)\rho_d]+\Theta(|\vec{x}|^3)
\end{align*}

Choosing a curve in $\SU(n)$ passing through $U_0$ such that
$\xs=\lambda_{k\ell}t$, we have
\begin{eqnarray*}
J&=&\Tr[U_0\rho_d U_0^\dagger\rho_d]+ t^2
\{\Tr(\lambda_{k\ell}U_0\rho_d U_0^\dagger
\lambda_{k\ell}^\dagger\rho_d)\\
&&-\Tr(U_0\rho_d U_0^\dagger\rho_d)\} +\Theta(|t|^3).
\end{eqnarray*}
Analogously, choosing a curve in $SU(n)$ passing through $U_0$ such
that $\xs=\bar\lambda_{k\ell}t$, we have
\begin{eqnarray*}
J&=&\Tr[U_0\rho_d U_0^\dagger\rho_d]+ t^2
\{\Tr(\bar\lambda_{k\ell}U_0\rho_d U_0^\dagger
\bar\lambda_{k\ell}^\dagger\rho_d)\\
&&-\Tr(U_0\rho_d U_0^\dagger\rho_d)\} +\Theta(|t|^3).
\end{eqnarray*}

The conjugate action of $\lambda_{k\ell}$ or $\bar\lambda_{k\ell}$
on the critical point $\rho_0=U_0\rho_d U_0^\dagger$ swaps the
$k$-th and $\ell$-th diagonal elements. Since $\rho_d$ is
non-degenerate, any swap $\lambda_{k\ell}$ or $\bar\lambda_{k\ell}$
will either increase or decrease the value of $J=\sum_{k=1}^n w_k
w_{\tau(k)}$, corresponding to a minimum or maximum along that
direction. This holds for all directions of $\lambda_{k\ell}$ and
$\bar\lambda_{k\ell}$, and since the dimensions of $\mathcal{T}$ and
$\mathcal{M}$ are both $n^2-n$, we have found $n^2-n$ independent
directions along which $J$ corresponds to a maximum or minimum. Thus
these $n!$ critical points are all hyperbolic. The maximal critical
value occurs only when $\rho_0=\rho_d$ and the minimal value occurs
only when $w_{\tau(k)}$'s are in an increasing order.  For all other
critical values, there always exists a swap that will increase the
value of $J$ and a swap that will decrease it, showing that they are
saddle points of the $J$.
\end{proof}

\section{Lyapunov Control under an ideal Hamiltonian}
\label{sec:conv_ideal}

In this section we consider the implications of the results of the
previous sections on the convergence behaviour and effectiveness of
Lyapunov control of a quantum system under an ideal Hamiltonian, i.e.,
assuming $H_0$ is strongly regular and $H_1$ is off-diagonal and fully
connected. Without loss of generality we can also assume $H_0\in\su(n)$,
as the identity part of $H_0$ only changes the global phase.  Once the
form of the Hamiltonian is fixed, the LaSalle invariant set $E$ depends
on the target state $\rho_d$ only.  We discuss in detail the two most
important cases when (a) $\rho_d$ is a pseudo-pure state and hence
$\dim(\M)=2n-2$, and when (b) $\rho_d$ is generic and $\dim(\M)=n^2-n$,
and conclude with a brief discussion of degenerate stationary target
states $\rho_d$.

\subsection{Pseudo-pure state control}

In this section we consider the special class of density operators
acting on $\H=\CC^n$ whose spectrum consists of two eigenvalues
$\{w,u\}$ where $u=(1-w)/(n-1)$ occurs with multiplicity $n-1$,
which includes pure states with spectrum $\{1,0\}$.  We first consider
the special case of a two-level system as the results for this case 
can be easily visualized in $\RR^3$ and are useful in the general 
discussion of pure-state control problems for $n$-level systems that
follows.

\subsubsection{Two-level systems}

For a two-level system strong regularity of $H_0$ simply means that the
energy levels are non-degenerate and full connectivity of $H_1$ requires
only $b_{12}\neq 0$, conditions that are satisfied in all but trivial 
cases.  The density operator of a two-level system can be written as
\begin{equation}
 \rho = \frac{1}{2} \left(\sigma_0 + x \sigma_x + y \sigma_y + z \sigma_z\right),
\end{equation}
where $\vec{s}=(x,y,z)\in\RR^3$ and the Pauli matrices are
\begin{equation*}
  \sigma_0 = \begin{bmatrix} 1 & 0 \\ 0 & 1\end{bmatrix},
  \sigma_x = \begin{bmatrix} 0 & 1 \\ 1 & 0\end{bmatrix},
  \sigma_y = \begin{bmatrix} 0 & -i \\ i & 0\end{bmatrix},
  \sigma_z = \begin{bmatrix} 1 & 0 \\ 0 & -1\end{bmatrix}.
\end{equation*}
Noting that $\Tr(\rho^2)=\frac{1}{2}(1+\norm{\vec{s}}^2)$ shows that
in this representation pure states, characterized by $\Tr(\rho^2)=1$,
correspond to points on the surface of the unit sphere $S^2\in\RR^3$,
while mixed states ($\Tr(\rho^2)<1$) correspond to points in the interior.
The vector $\vec{s}$ is often called the Bloch vector of the quantum
state.  Any unitary evolution of $\rho(t)$ under a constant Hamiltonian
corresponds to a rotation of $\vec{s}(t)$ about a fixed axis in $\RR^3$,
and free evolution under $H_0=\diag(a_1,a_2)$ in particular corresponds
to a rotation of the Bloch vector $\vec{s}(t)$ about the $z$-axis.  Thus,
in this special case the path $\vec{s}(t)$ traced out by any Bloch vector
$\vec{s}_0$ evolving under any constant Hamiltonian forms a circle, i.e.,
a closed periodic orbit.

Let $\vec{s}=(x,y,z)$ and $\vec{s}_d=(x_d,y_d,z_d)$ be the Bloch
vectors of $\rho$ and $\rho_d$, respectively.  It is straightforward
to show that $[\rho,\rho_d]$ diagonal implies
\begin{equation}
\label{eq:2D}
  z x_d-xz_d = 0, \qquad
  y z_d-zy_d = 0.
\end{equation}

If (a) $z_d\neq 0$ then $(x,y)= \alpha(x_d,y_d)$ with $\alpha=z/z_d$,
and thus $x^2+y^2+z^2=\alpha^2(x_d^2+y_d^2+z_d^2)$, and the RHS has
to equal $x_d^2+y_d^2+z_d^2$ since $\norm{\vec{s}}=\norm{\vec{s}_d}$.
Thus $\alpha=\pm 1$ and $(x,y,z)=\pm(x_d,y_d,z_d)$, and the
corresponding density operators $\rho,\rho_d$ commute, $[\rho,\rho_d]=0$.

If (b) $z_d=0$ then either (b1) $x_d=0$ and $y_d=0$ or (b2) $z=0$.
In case (b1) we have $\vec{s}_d=(0,0,0)$, i.e., the target state is
the completely mixed state.  Since the completely mixed state forms
a trivial equivalence class under unitary evolution, the invariant
set in this case is $E=\{(\vec{0},\vec{0})\}$.  Case (b2) is more
interesting with the invariant set being
\begin{equation}
 E=\{(\vec{s},\vec{s}_d): z=z_d=0, x^2+y^2=x_d^2+y_d^2\},
\end{equation}
i.e., all pairs of Bloch vectors that lie on a circle of radius
$\norm{\vec{s}_d}$ in the equatorial plane.  Notice that this set is
significantly larger than the set of critical points of $V$, which
consists only of $\{\pm\vec{s}_d\}$.

\begin{figure}
\scalebox{0.75}{\includegraphics{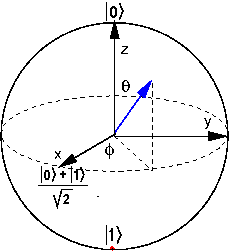}}
\caption{Bloch sphere: There is a one-to-one correspondence
between states $\rho$ of a two-level quantum system and points
inside the Bloch ball.  Pure states correspond to points on the
surface of the Bloch ball, mixed states to points in the interior.
The Bloch vector of a pure state
$\ket{\Psi}=\cos\theta\ket{0}+e^{i\phi} \sin\theta \ket{1}$
is $\vec{s}=(\sin(2\theta)\cos(\phi),\sin(2\theta)\sin\phi,\cos(2\theta))$.}
\label{fig:Bloch}
\end{figure}

Hence, the invariant set $E$ depends on the choice of the target
state $\vec{s}_d(t)$.  Ignoring the trivial case (b1), if
$\vec{s}_d(t)$ is not in the equatorial plane~\footnote{$\vec{s}_d(t)$
is lies in the equatorial plane for any $t$ then it lies in the
equatorial plane for all $t$ since its free evolution corresponds
to a rotation about the $z$-axis.}
then the invariant set is equal to the set of critical points
$\{\pm \vec{s}_d(t)\}$ of $V$, and hence $\vec{s}(t)$ for a given
initial state $\vec{s}(0)$ will converge to either $\vec{s}_d(t)$
or its antipodal point $-\vec{s}_d(t)$.  Furthermore, since
$V(\vec{s}(t),\vec{s}_d(t))$ assumes its (global) maximum for
$\vec{s}(t)=-\vec{s}_d(t)$ and $V$ is non-increasing, $\vec{s}(t)$
will converge to the target trajectory $\vec{s}_d(t)$ for all
initial states $\vec{s}(0)\neq-\vec{s}_d(0)$.

If $\vec{s}_d(t)$ lies in the equatorial plane $z_d=0$ then the
invariant set consists of all points $(\vec{s}(t),\vec{s}_d(t))$
with $z(t)=z_d(t)=0$ and $\norm{\vec{s}(t)}=\norm{\vec{s}_d(t)}$,
which lie on a circle of radius $\norm{\vec{s}_d(0)}$ in the $z=0$
plane, and we can only say that any initial state
$\vec{s}(0)\not\in E$
will converge to a trajectory $\vec{s}_1(t)$ with $z_1(t)=0$ and
$\norm{\vec{s}_1(t)}=\norm{\vec{s}_d(t)}$.
$V(\vec{s}_1(t),\vec{s}_d(t))$ can take any limiting value between
$V_{\rm min}=0$ and $V_{\rm max}=2\norm{\vec{s}_d(0)}^2$ in this
case.  Notice that, although in almost all cases the trajectories
$\vec{s}(t)$ and $\vec{s}_d(t)$ remain a fixed, non-zero distance
apart for all times in this case, this result is consistent with
the results in~\cite{altafini1,altafini2} for the weaker notion of
orbit convergence, since the circle in the equatorial plane in this
case corresponds to the orbit of $\vec{s}_d(t)$ under $H_0$, and
any initial state converges to this set in the sense that the
distance of $\vec{s}(t)$ to \emph{some} point on this circle goes
to zero for $t\to\infty$.

\subsubsection{Pseudo-pure states for $n>2$}

The density operator $\rho$ for a pseudo-pure state in $\CC^n$ with 
spectrum $\{w,u\}$ can be written as
\begin{equation}
\label{eq:pseudo-pure}
  \rho=w \Pi + \frac{1-w}{n-1}\Pi^\perp, \quad 0<w \le 1,
\end{equation}
where $\Pi$ is a rank-1 projector.  Since $\Pi+\Pi^\perp=I$, we have 
$0=[x,\Pi+\Pi^\perp]$ for all $x$, and thus $[x,\Pi^\perp]=-[x,\Pi]$. 
If $\rho_d(0)$ is pseudo-pure, with $\rho_d(0)=w\Pi_0+u\Pi_0^\perp$, 
then for any $(\rho_1,\rho_2)\in E$, $\rho_1$ and $\rho_2$ must also be
pseudo-pure, with the same spectrum $\{w,u\}$, i.e.,
$\rho_k=w\Pi_k+u\Pi_k^\perp$ for $k=1,2$. We have
\begin{align}
 [\rho_1,\rho_2]
 &= w^2[\Pi_1,\Pi_2] + uw[\Pi_1^\perp,\Pi_2] \nonumber\\
 &\qquad\qquad\qquad+ uw[\Pi_1,\Pi_2^\perp] +u^2[\Pi_1^\perp,\Pi_2^\perp] \nonumber\\
 &= w^2[\Pi_1,\Pi_2] - 2 uw[\Pi_1,\Pi_2] +u^2[\Pi_1,\Pi_2] \nonumber\\
 &= (w-u)^2 [\Pi_1,\Pi_2]. \label{eq:pseudo-pure:comm}
\end{align}
Thus the LaSalle invariant set contains all points such that
$M=[\Pi_1,\Pi_2]$ is diagonal. Since $\Pi_k$, $k=0,1,2$, are
rank-$1$ projectors, $\Pi_k=\ket{\Psi_k}\bra{\Psi_k}$, where
$\ket{\Psi_k}$ are unit vectors in $\CC^n$.  Setting
\begin{equation}
\begin{split}
 \ket{\Psi_1} &=(a_1e^{i\alpha_1},\ldots,a_ne^{i\alpha_n})^T \\
 \ket{\Psi_2} &=(b_1e^{i\beta_1},\ldots,b_ne^{i\beta_n})^T
\end{split}
\end{equation}
where $\ket{\Psi_k}$, $k=0,1,2$, are pure states, represented as
unit vectors in $\CC_+^n$. We have
\begin{align*}
M&=[\Pi_1,\Pi_2] \\
&=\ket{\Psi_1}\langle\Psi_1|\Psi_2
\rangle\bra{\Psi_2}-\ket{\Psi_2}\langle\Psi_2|\Psi_1
\rangle\bra{\Psi_1}
\end{align*}
For $(\rho_1,\rho_2)\in E$, we require that all off-diagonal
elements equal to zero, i.e.:
\begin{equation}
\label{eq:Mkl}
  M_{k\ell} = a_k b_\ell e^{i(\alpha_k-\beta_\ell)} \ip{\Psi_1}{\Psi_2}
             -a_\ell b_k e^{i(\beta_k-\alpha_\ell)} \ip{\Psi_2}{\Psi_1}.
\end{equation}
for all $k\ne \ell$. Let $\ip{\Psi_1}{\Psi_2}= r e^{i\theta}$. We
have the following two cases.

(a) $r=0$ i.e $\ip{\Psi_1}{\Psi_2}=0$ or $\rho_1 \perp \rho_2$. In
this case, $[\rho_1,\rho_2]=0$, and
$V(\rho_1,\rho_2)=V_{\rm max}=(w-u)^2$.

(b) If $r\neq 0$ then Eq.~(\ref{eq:Mkl}) together with $M_{kk}=0$
leads to $n(n-1)/2$ non-trivial equations for the population and
phase coefficients, respectively:
\begin{align}
 a_k b_\ell         &= a_\ell b_k, \label{eq:Mkla}\\
 \beta_k+\beta_\ell &= \alpha_k+\alpha_\ell + 2\theta \label{eq:Mklb}.
\end{align}
If $a_k=0$ then $0=a_k b_\ell=a_\ell b_k$ for $\ell\neq k$ and thus
we must have $b_k=0$ as $a_\ell=0$ for all $\ell$ is not allowed as
$\vec{a}$ is a unit vector.  Ditto for $b_k=0$.  Let $I_+$ be the
set of all indices $k$ so that $a_k,b_k\neq0$.  Then the remaining
non-trivial equations for the population coefficients can be rewritten
\begin{equation}
   \frac{a_k}{b_k} = \frac{a_\ell}{b_\ell}, \qquad \forall k,\ell \in I_+
\end{equation}
and thus $\vec{a}=\gamma \vec{b}$ and as $\vec{a}$ and $\vec{b}$
are unit vectors in $\RR_+^n$, $\gamma=1$ and $\vec{a}=\vec{b}$.

As for the phase equations~(\ref{eq:Mklb}), if $a_k=b_k=0$ then
$M_{k\ell}=0$ is automatically satisfied, thus the only non-trivial
equations are those for $k,\ell\in I_+$.  If the set $I_+$ contains
$n_1>2$ indices then taking pairwise differences of the $n_1(n_1-1)/2$
non-trivial phase equations and fixing the global phase of $\ket{\Psi_k}$
by setting $\alpha_{n_1}=\beta_{n_1}=0$ shows that
$\vec{\alpha}=\vec{\beta}$.  For example, suppose $I_+=\{1,2,3\}$
then we have $3$ non-trivial phase equations
\begin{align*}
  \beta_1+\beta_2 & = \alpha_1+\alpha_2 + 2\theta, \\
  \beta_1+\beta_3 & = \alpha_1+\alpha_3 + 2\theta, \\
  \beta_2+\beta_3 & = \alpha_2+\alpha_3 + 2\theta,
\end{align*}
taking pairwise differences gives
\begin{align*}
  \beta_2-\beta_3 & = \alpha_2-\alpha_3,\\
  \beta_1-\beta_3 & = \alpha_1-\alpha_3,\\
  \beta_1-\beta_2 & = \alpha_1-\alpha_2,
\end{align*}
and setting $\alpha_3=\beta_3=0$ shows that we must have
$\alpha_2=\beta_2$ and $\alpha_3=\beta_3$.  Thus, together with
$\vec{a}=\vec{b}$ we have $\rho_1=\rho_2$.  If $I_+$ contains only
a single element then $\ket{\Psi_1}$ and $\ket{\Psi_2}$ differ at
most by a global phase and again $\rho_1=\rho_2$ follows.
Incidentally, note that for $\ket{\Psi_1}=\ket{\Psi_2}$ we have
$\ip{\Psi_1}{\Psi_2}=1$, i.e., $r=1$, $\theta=0$.

The only exceptional case arises when $I_+$ contains exactly two
elements, say $\{1,2\}$, as in this case there is only a single
phase equation $\beta_1+\beta_2=\alpha_1+\alpha_2+2\theta$, and
thus even fixing the global phase by setting $\alpha_2=\beta_2=0$,
only yields $\beta_1-\alpha_1=2\theta$.  This combined with
$\vec{a}=\vec{b}$ gives
\begin{gather*}
  r e^{i\theta} = \ip{\Psi_1}{\Psi_2}
                = a_1^2 e^{2i\theta} + a_2^2
\end{gather*}
and thus $a_1^2 e^{i\theta} + a_2^2 e^{-i\theta} = r$ or
\begin{equation*}
   2i\sin\theta (a_1^2-a_2^2) = 0.
\end{equation*}
Therefore, either $\theta=0$ or $a_1=a_2$.  If $\theta=0$ then
$\vec{\alpha}=\vec{\beta}$ and $\rho_1=\rho_2$, which is one
possible solution in the LaSalle invariant set. If $\theta\neq 0$,
then any $(\rho_1,\rho_2)$ satisfying
\begin{subequations}
\label{eqn:pseudo}
\begin{align}
\ket{\Psi_1}&=2^{-1/2}(1,e^{i\alpha},0,\ldots,0)^T\\
\ket{\Psi_2}&=2^{-1/2}(1,e^{i\beta},0,\ldots,0)^T
\end{align}
\end{subequations}
with $\beta-\alpha=2\theta$ is also in the LaSalle invariant set.  

Hence, if the target state is $\rho_d(0)=w\Pi_0+u\Pi_0^\perp$ with
$\Pi_0=\ket{\Psi_0}\bra{\Psi_0}$ and $\ket{\Psi_0}$ has only two 
nonzero components with equal norm, e.g., if 
\begin{equation}
  \label{eq:rho_limit:t}
  \rho_d(0) =
  \begin{pmatrix} r_{11} & r_{12}(t)& 0 & \ldots & 0\\
                  r_{12}(t)^\dagger & r_{11} & 0 & \ldots & 0\\
                  0        & 0      & u & \\
                  \vdots   & \vdots &        & \ddots \\
                  0        & 0      &        &        & u
  \end{pmatrix},
\end{equation}
with $r_{11}=\frac{1}{2}(w+u)$, $r_{12}(t)=\frac{1}{2}(w-u)e^{i\alpha}$,
and $\ket{\Psi_0}=2^{1/2}(1,e^{i\alpha},0,\ldots,0)^T$, then the invariant 
set contains all points $(\rho_1,\rho_2)$ satisfying~(\ref{eqn:pseudo}), 
which includes $\rho_1=\rho_2$ and $\rho_1\perp\rho_2$.  Since $\rho_1$ 
and $\rho_2$ lie on the orbit of $\rho_d(0)$, any solution $\rho(t)$ will 
converge to this orbit but we \emph{cannot} guarantee $\rho(t)\to\rho_d(t)$ 
as $t\to +\infty$.  This case is analoguous to the case where the target
state was located in the equatorial plane of the Bloch ball for $n=2$.  
For all other $\rho_d(0)$ the LaSalle invariant set contains only points
with either $\rho_1=\rho_2$ or $\rho_1\perp\rho_2$, corresponding to $V=0$ 
and $V=V_{\rm max}$, respectively, and since $V$ is non-increasing, any 
solution $\rho(t)$ with $V(\rho(0),\rho_d(0))<V_{\rm max}$ will converge 
to $\rho_d(t)$ as $t\to +\infty$.

In summary we have the following result:
\begin{theorem}
\label{thm:conv:ideal-pseudo-pure}
Given a pseudo-pure state target state $\rho_d(t)$ with spectrum $\{w,u\}$ 
and `ideal' Hamiltonians as defined, Lyapunov control is effective, i.e., 
any solution $\rho(t)$ with $V(\rho(0),\rho_d(0))<V_{\rm max}$ will converge 
to $\rho_d(t)$ as $t\to+\infty$, \emph{except} when $\rho_d$ has a single 
pair of non-zero off-diagonal entries of the form 
$r_{k\ell}(t)=\frac{1}{2}(w-u)e^{i\alpha}$ and 
$r_{kk}=r_{\ell\ell}=\frac{1}{2}(w+u)$.  In the latter case any solution 
$\rho(t)$ will converge to the orbit of $\rho_d(t)$ but in general 
$\rho(t)\not\to\rho_d(t)$ as $t\to +\infty$ and $V(\rho,\rho_d)$ can take 
any limiting value between $0$ and $V_{\rm max}$.
\end{theorem}

\subsection{Generic-state Control}

For generic states $\rho_d$ we shall distinguish between stationary and 
non-stationary target states.  Recall that $\rho_d(t)$ is stationary if 
and only if $[H_0,\rho_d(0)]=0$.  If $H_0$ has non-zero eigenvalues, 
which is always the case if $H_0$ is strongly regular, then this happens 
if and only if $\rho_d$ is diagonal in the eigenbasis of $H_0$.

\subsubsection{Generic stationary target state}

When $\rho_d$ is a stationary state Eq.~(\ref{eqn:auto}) can be reduced 
to a dynamical system on $\M$
\begin{subequations}
\label{eqn:auto1}
\begin{align}
\dot \rho(t) &=-i [ H_0+f(\rho)H_1, \rho(t) ]\\
f(\rho)&=\Tr([-iH_1,\rho(t)]\rho_d)
\end{align}
\end{subequations}
and the LaSalle invariant set can be reduced to
\begin{align}
  E &=\{\rho_0|\dot V(\rho(t))=0,\rho(0)=\rho_0\} \nonumber\\
    &=\{\rho_0: [\rho_0,\rho_d]=\diag(c_1,\ldots,c_n)\}
\end{align}
according to Theorem~\ref{thm:lasalle:4}. 

If $\rho_d$ is generic and both $\rho_d$ and $[\rho,\rho_d]$ 
are diagonal then $\rho$ must be diagonal and $[\rho,\rho_d]=0$ since 
suppose $\rho_d=\diag(w_1,\ldots,w_n)$ and $\rho=(r_{k\ell})$.  Then 
the $(k,\ell)$-th component of $[\rho_d,\rho]$ is $r_{k\ell}(w_k-w_\ell)$.
Since $\rho_d$ is generic $w_k\neq w_\ell$ except for $k=\ell$ and thus
$[\rho_d,\rho]$ is diagonal only if $r_{k\ell}=0$ for $k\ne \ell$, i.e.,
if $\rho$ is diagonal.  Since the commutator of two diagonal matrices
vanishes, the invariant set in this case reduces to the set of all
$\rho_0$ that commute with the stationary state $\rho_d$, i.e., in 
this case the invariant set $E$ not only contains the set of critical
points $F$ of the Lyapunov function but we have $E=F$.  In summary we
have:

\begin{theorem}
\label{thm:generic:crit}
If $\rho_d$ is a generic stationary target state then the invariant set 
$E$ contains exactly the $n!$ critical points of the Lyapunov function 
$V$, i.e., the stationary states $\rho_d^{(k)}$, $k=1,\ldots,n!$, that 
commute with $\rho_d$ and have the same spectrum.  
\end{theorem}

%Without loss of generality assume the diagonal elements of $\rho_d$ are 
%arranged in decreasing order.  Each $\rho_d^{(k)}$ is also diagonal and 
%its diagonal elements are a permutation of those of $\rho_d$.  Clearly,
%$\rho_d^{(0)}=\rho_d$ and the operator $\rho_d^{(n!)}$ whose eigenvalues
%are in arranged in increasing order correspond to the points with $V=0$
%and $V=V_{\rm max}$, respectively.

These critical points are the only stationary solutions and all the
other solutions must converge to \emph{one} of these points.  However,
we still cannot conclude that all or even most solutions converge to 
the target state $\rho_d$.  In fact we shall see that not all solutions 
$\rho(t)$ converge to $\rho_d$ even for $\rho(0)\not\in E$.  However, the 
target state $\rho_d$ is the only hyperbolic sink of the dynamical system, 
and all other critical points are hyperbolic saddles or sources, and 
therefore most (almost all) initial states will converge to the target 
state as desired.

We note that Theorem~\ref{thm:crit:2} guarantees that for a given generic 
stationary state $\rho_d$ the critical points of the Lyapunov function 
$V(\rho)$ are hyperbolic.  Thus, if the dynamical system was the gradient 
flow of $V(\rho)$ then asymptotic stability of these fixed points could 
be derived directly from the associated index number of the Morse function 
$V$\cite{Morse}.  However, since the dynamical system~(\ref{eqn:auto1}) is 
\emph{not} the gradient flow, further analysis of the linearization of the 
dynamics near the critical points is necessary.  To this end we require a
real representation for our complex dynamical system.  A natural choice is 
the Bloch vector (sometimes also called Stokes tensor) representation, 
where a density operator $\rho$ is represented as a vector 
$\vec{s}\in\RR^{n^2-1}$ defined by
$s_k=\Tr(\rho \xi_k)$, where $\xi_k=-i\sigma_k$ and $\{\sigma_k\}$ is
the orthonormal basis of $\su(N)$, as defined in the proof of
Theorem~\ref{thm:crit:2}.  The adjoint action $\Ad_{iH}(\rho)=[iH,\rho]$
in this basis is given by a real anti-symmetric matrix $A$ acting on
$\vec{s}$.  Therefore, the quantum dynamical system~(\ref{eqn:auto}) can
be equivalently represented as
\begin{subequations}
\begin{align*}
\dot {\vec{s}}(t)   &= (A_0+f(\vec{s},\vec{s}_d)A_1)\vec{s}(t)\\
\dot {\vec{s}}_d(t) &= A_0\vec{s}_d(t)\\
f(\vec{s},\vec{s}_d)&= \vec{s_d}^TA_1\vec{s},
\end{align*}
\end{subequations}
where $A_0=A_{-iH_0}$ and $A_1=A_{-iH_1}$.  For a fixed stationary target 
state $\rho_d$ this system can be reduced to
\begin{subequations}
\label{eq:sys_real}
\begin{align}
\dot {\vec{s}}(t) &= (A_0+f(\vec{s})A_1)\vec{s}(t)\\
        f(\vec{s})&= \vec{s_d}^TA_1\vec{s}.
\end{align}
\end{subequations}
The linearized system near the critical point $\vec{s}_0$ is 
\begin{equation}
\label{eqn:linear}
  \dot {\vec{s}}= D_f(\vec{s}_0)\cdot (\vec{s}-\vec{s}_0),
\end{equation}
where $D_f(\vec{s}_0)=A_0+A_1 \vec{s}_0\cdot \vec{s_d}^T A_1$ is a 
linear map defined on $\RR^{n^2-1}$.

The state space $S_\M$ of the real dynamical system is the set of all
Bloch vectors $\vec{s}\in \RR^{n^2-1}$ that correspond to density
operators $\rho\in\M$.  For generic states, $\M$ is the complex flag
manifold $\M \simeq \SU(n)/\exp(\C)$, where $\C$ is the Cartan subspace
of the Lie algebra $\su(n)$.  Hence, the tangent space $T_\M(\rho_0)$
of $\M$ at any point $\rho_0$ corresponds to the non-Cartan subspace
$\T$ of $\su(n)$ and the Cartan elements of $\su(n)$ correspond to the
tangent space of the isotropy subgroup of $\rho_0$.  In the equivalent
real representation $\RR^{n^2-1}$ is therefore the direct sum of the 
($n^2-n$)-dimensional tangent space $S_\T$ to the manifold $S_\M$ and 
the ($n-1$)-dimensional subspace $S_\C$ corresponding to the Cartan
subspace of $\su(n)$.

\begin{theorem}
\label{thm:generic:hyperbolic}
For a generic stationary target state $\rho_d$ all the critical points 
of the dynamical system~(\ref{eqn:auto1}) are hyperbolic. $\rho_d$ is 
the only sink, all other critical points are saddles, except the global 
maximum, which is a source.  
\end{theorem}

\begin{proof}
We show that the critical points $\vec{s}_0$ of the corresponding real 
dynamical system~(\ref{eq:sys_real}) defined on $S_\M$ are hyperbolic.  
To this end, we first show that $D_f(\vec{s}_0)$ vanishes on the 
$(n-1)$-dimensional subspace $S_\C$, which is orthogonal to the tangent 
space of $S_\M$.  In the second step we show that the restriction of 
$D_f(\vec{s}_0)$ onto the tangent space of $S_\M$ is well-defined and 
has $n^2-n$ non-zero eigenvalues.  Finally, we show that the restriction 
of $D_f(\vec{s}_0)$ onto the tangent space of $S_\M$ does not have any
purely imaginary eigenvalues, from which it follows that $\vec{s}_0$ 
is a hyperbolic fixed point of the (real) dynamical system defined on 
$S_\M$, and the local behavior of the original dynamical system can 
therefore be approximated by the linearized system~\cite{stability}.

\begin{lemma}
$D_f(\vec{s}_0)$ vanishes on the subspace $S_\C$.
\end{lemma}

\begin{proof}
To show that $D_f(\vec{s}_0)\vec{s}=0$ for all $\vec{s}\in S_\C$, it 
suffices to show that $A_0 \vec{s}=0$ and $\vec{s}_d^T A_1 \vec{s}=0$
for $\vec{s} \in S_\C$.  $\vec{s}\in S_\C$ corresponds to density 
operators $\rho\in i\C$, i.e., $\rho$ diagonal.  As $A_0\vec{s}$ is 
the Bloch vector associated with $[-iH_0,\rho]$, $-iH_0$ is diagonal 
and since diagonal matrices commute, $[-iH_0,\rho]=0$ and $A_0\vec{s}=0$ 
follows immediately.  To establish the second part, we note that for 
$i\rho \in\C$ and $-iH_1\in \T$, we have $[-iH_1,i\rho] \in \T$, or 
$[-iH_1,\rho]\in i\T$, and $A_1\vec{s} \in S_\T$.  Since $\rho_d$ is 
diagonal and thus $\vec{s}_d \in S_\C \perp S_\T$, we have 
$\vec{s}_d^TA_1\vec{s}=0$ for $\vec{s}\in S_\C$.
\end{proof}

This lemma shows that $\vec{s}_0$ is not a hyperbolic fixed point of the
dynamical system~(\ref{eq:sys_real}) defined on $\RR^{n^2-1}$.  However,
we are only interested in the dynamics on the manifold $S_\M$, and thus
it suffices to show that $\vec{s}_0$ is a hyperbolic fixed point of the
restriction of $D_f(\vec{s}_0)$ to the tangent space $S_\T$ of $S_\M$.

\begin{lemma}
The restriction $B$ of $D_f(\vec{s}_0)$ to $S_\T$ is well-defined and 
has $n^2-n$ non-zero eigenvalues.
\end{lemma}

\begin{proof}
Since we already know that $S_\C$ is in the kernel of $D_f(\vec{s}_0)$,
it suffices to show that the image of $D_f(\vec{s}_0)$ is contained in
$S_\T$, i.e., $D_f(\vec{s}_0)\vec{s}\in S_\T$.  To this end
\begin{align*}
  D_f(\vec{s}_0) \vec{s}
  &= A_0 \vec{s} + A_1 \vec{s}_0  \; \vec{s_d}^T A_1\vec{s} \\
  &= A_0 \vec{s} + (\vec{s_d}^T A_1 \vec{s} ) \, A_1 \vec{s}_0
\end{align*}
shows that it suffices to show that $A_0\vec{s} \in S_\T$ and
$A_1\vec{s}_0\in S_\T$.  Both relations follow from the fact that the
commutator of a Cartan element and a non-Cartan element of the Lie
algebra $\su(n)$ is always in the non-Cartan algebra $\T$, and thus
$[-iH_0,\rho] \in i\T$ since $-iH_0\in \C$, and $[-iH_1,\rho_d]\in i\T$
since $i\rho_d \in \C$.  Therefore, the restriction $B: S_\T \to S_\T$ 
of $D_f(\vec{s}_0) \vec{s}$ is well defined.

Furthermore, the restriction of $A_0$ to $S_\T$ is a block-diagonal 
matrix $B_0=\diag(A_0^{(k,\ell)})$ with
\begin{equation*}
  A_0^{(k,\ell)}= \omega_{k\ell} 
  \begin{pmatrix} 0 & 1 \\ -1 & 0 \end{pmatrix}.
\end{equation*}
The restriction $\vec{u}$ of $A_1\vec{s}_0$ to $S_\T$ is a column vector 
$(\vec{u}^{(1,2)};\vec{u}^{(1,3)};\ldots;\vec{u}^{(n-1,n)})$ of length 
$n(n-1)$ consisting of $n(n-1)/2$ elementary parts
\begin{equation}
 \label{eq:ukl}
 \vec{u}^{(k,\ell)} =
\frac{\Delta_{\tau(k)\tau(\ell)}}{\sqrt{2}}
 \begin{pmatrix}
 \Im(b_{k\ell}) \\ \Re(b_{k\ell})
\end{pmatrix}
\end{equation}
for $k=1,\ldots,n-1$ and $\ell=k+1,\ldots,n$.  Similarly, let
$\vec{v}$ be the restriction of $A_1\vec{s}_d$ to $S_\T$.  Then
$\vec{v}=(\vec{v}^{(1,2)};\ldots;\vec{v}^{(n-1,n)})$ with
$\vec{v}^{(k,\ell)}$ as in Eq.~(\ref{eq:ukl}) and $\tau$ the
identity permutation.

Thus the restriction of $D_f(\vec{s}_0)$ to the subspace $S_\T$ is
$B=B_0-\vec{u}\vec{v}^T$.  Since $\omega_{k\ell}\neq 0$ for all
$k,\ell$ by regularity of $H_0$, we have $\det(B_0) =
\prod_{k,\ell} \omega_{k\ell}^2 \neq 0$, i.e., $B_0$ invertible,
and by the matrix determinant lemma~\cite{matrix}
\begin{equation*}
  \det(B) = \det(B_0 - \vec{u}\vec{v}^T)
  = (1-\vec{v}^T B_0^{-1} \vec{u}) \det(B_0).
\end{equation*}
$B_0^{-1}$ is a block-diagonal matrix with blocks
\begin{equation*}
 C^{(k,\ell)} = [A_0^{(k,\ell)}]^{-1}
              = \frac{1}{\omega_{k\ell}}
 \begin{pmatrix} 0 & -1 \\ 1 & 0 \end{pmatrix}.
\end{equation*}
Hence $\vec{v}^T B_0^{-1}\vec{u}=\sum_{k,\ell}
[\vec{v}^{(k,\ell)}]^T C^{(k,\ell)} \vec{u}^{(k,\ell)}$ vanishes
since
\begin{equation*}
  (\Im(b_{k\ell}),\Re(b_{k\ell})
 \begin{pmatrix} 0 & -1 \\ 1 & 0 \end{pmatrix}
 \begin{pmatrix}
 \Im(b_{k\ell}) \\ \Re(b_{k\ell})
\end{pmatrix} = 0, \quad \forall k, \ell.
\end{equation*}
Therefore, $\det(B)=\det(B_0)\neq0$ and thus the restriction of
$D_f(\vec{s}_0)$ to $S_\T$ is invertible, and hence has only
non-zero eigenvalues.
\end{proof}

\begin{lemma}
\label{lemma:generic:3}
If $i\beta$ is a purely imaginary eigenvalue of $B$ then it must be 
an eigenvalue of $B_0$, i.e., $i\beta=\pm i\omega_{k\ell}$ for some
$(k,\ell)$, and either the associated eigenvector $\vec{e}$ must be 
an eigenvector of $B_0$ with the same eigenvalue, or the restriction 
of $A_1\vec{s}_0$ to the $(k,\ell)$ subspace must vanish.
\end{lemma}

\begin{proof}
If $i\gamma$ is not an eigenvalue of $B_0$ then $(B_0-i\beta I)$ is 
invertible and by the matrix determinant lemma
\begin{align*}
  0 &= \det(B_0 - \vec{u}\vec{v}^T - i\beta I) \\
    &= \det( (B_0-i\beta I) - \vec{u}\vec{v}^T) \\
    &= (1-\vec{v}^T (B_0-i\beta I)^{-1} \vec{u}) \det(B_0-i\beta I).
\end{align*}
Since $\det(B_0-i\beta I)\neq 0$ we must therefore have
\begin{equation*}
  \vec{v}^T (B_0-i\beta I)^{-1} \vec{u} = 1.
\end{equation*}
Noting that $(B_0-i\beta I)^{-1}$ is block-diagonal with blocks
\begin{equation}
 \label{eq:B0inv}
  C_\beta^{(k,\ell)} =
  \frac{1}{\omega_{k\ell}^2-\beta^2}
  \begin{pmatrix}
   -i\beta & -\omega_{k\ell} \\ \omega_{k\ell} & -i\beta
  \end{pmatrix},
\end{equation}
\begin{equation*}
  \Big(\Im(b_{k\ell}),\Re(b_{k\ell}\Big)
 \begin{pmatrix} -i\beta & -\omega_{k\ell} \\ \omega_{k\ell} & -i\beta \end{pmatrix}
 \begin{pmatrix}
 \Im(b_{k\ell}) \\ \Re(b_{k\ell})
\end{pmatrix} = -i\beta |b_{k\ell}|^2
\end{equation*}
for all $k,\ell$, this leads to
\begin{align*}
 1& =\vec{v}^T (B_0-i\beta I)^{-1}\vec{u}
    =\sum_{k,\ell}[\vec{v}^{(k,\ell)}]^T C_\beta^{(k,\ell)}\vec{u}^{(k,\ell)} \\
  & = \frac{-i\beta}{2} \sum_{k,\ell}
    \frac{\Delta_{k\ell}\Delta_{\tau(k)\tau(\ell)}}{\omega_{k\ell}^2-\beta^2}
    |b_{k\ell}|^2.
\end{align*}
Since all terms in the sum are real this is a contradiction.  Thus if
$i\beta$ is a purely imaginary eigenvalue of $B$ then it must be an
eigenvalue of $B_0$.  

Since the spectrum of $B_0$ is $\{\pm i\omega_{k\ell}\}$, this means
$i\beta =\pm i\omega_{k\ell}$ for some $(k,\ell)$.  Without loss of 
generality assume $\gamma=\omega_{12}>0$ and let $\vec{e}=\vec{x}+i\vec{y}$ 
be the associated eigenvector of $B$.  Then
\begin{align}
\label{eqn:Bx0}
  B\vec{e} = (B_0-\vec{u}\vec{v}^T)(\vec{x}+i\vec{y})
           = i\omega_{12}(\vec{x}+i\vec{y}),
\end{align}
which is equivalent to
\begin{subequations}
\label{eq:Bx}
\begin{align}
 (B_0-\vec{u}\vec{v}^T)\vec{x}&= -\omega_{12} \vec{y}\\
 (B_0-\vec{u}\vec{v}^T)\vec{y}&=  \omega_{12} \vec{x}.
\end{align}
\end{subequations}
Multiplying~(\ref{eq:Bx}b) by $-\omega_{12} B_0^{-1}$ and adding it to 
(\ref{eq:Bx}a) 
\begin{align*}
 \underline{B_0\vec{x}}-\vec{u}\vec{v}^T\vec{x}
 +\omega_{12}B_0^{-1}\vec{u}\vec{v}^T\vec{y}
 &= \underline{-\omega_{12}^2 B_0^{-1} \vec{x}}
\end{align*}
Eq.~(\ref{eq:B0inv}) shows that
$-\omega_{12}^2[B_0^{(1,2)}]^{-1}=B_0^{(1,2)}$, i.e., on the $\T_{12}$
subspace the underlined terms above cancel, and thus the first two rows
of the above system of equations are
\begin{align*}
 \begin{pmatrix} u_1 \\ u_2 \end{pmatrix}
 (\vec{v}^T\vec{x})
= \begin{pmatrix} 0 & -1\\ 1 & 0 \end{pmatrix}
 \begin{pmatrix} u_1 \\ u_2 \end{pmatrix}
 (\vec{v}^T\vec{y}). 
\end{align*}
If $\vec{v}^T\vec{x}\neq 0$ then the last equation gives $u_1=-c^2u_1$
and $u_2=-c^2u_2$ with $c=\vec{v}^T\vec{y}/\vec{v}^T\vec{x}$, which can
only be satisfied if $u_1=u_2=0$.  Similarly if $\vec{v}^T\vec{y}\neq0$.
If $\vec{v}^T\vec{x}=\vec{v}^T\vec{y}=0$ then we have
$B\vec{e}=B_0\vec{e}=i\omega_{12}\vec{e}$, implying that $\vec{e}$ is 
an eigenvector of $B_0$ associated with $i\omega_{12}$. 
\end{proof}

The previous lemma shows that $B$ can have a purely imaginary eigenvalue
$i\beta$ only if $i\beta=\pm i\omega_{k\ell}$ for some $(k,\ell)$, and
either $\vec{u}^{(k,\ell)}=\vec{0}$, i.e., the projection of
$A_1\vec{s}_0$ onto the $(k,\ell)$ subspace vanishes, or the associated
eigenvector is also an eigenvector of $B_0$.  In the first case this 
means that $A_1\vec{s}_0$ vanishes on the subspace $\T_{k\ell}$, or 
equivalently that $[-iH_1,\rho_0]$ has no support in $\T_{k\ell}$, which 
contradicts the assumption that $H_1$ is fully connected and $\rho_0$ 
has non-degenerate eigenvalues.  On the other hand, if $\vec{e}$ is an 
eigenvector of $B_0$ with eigenvalue $i\beta=\pm i\omega_{k\ell}$ and 
$H_0$ is strongly regular then the projection of $\vec{e}$ onto the 
$(k,\ell)$ subspace is proportional to $(1,\pm i)$ and $\vec{e}$ is zero 
elsewhere, and thus $\vec{v}^T\vec{e}=0$ implies $\vec{v}^{(k,\ell)}=0$, 
which contradicts the fact that the projection $A_1\vec{s}_d$ or 
$[-iH_1,\rho_d]$ onto the $(k,\ell)$ subspace must not vanish if $H_1$ 
is fully connected and $\rho_d$ has non-degenerate eigenvalues.  Thus 
we can conlude that if $H_0$ is strongly regular, $H_1$ fully connected 
and $\rho_d$ has non-degenerate eigenvalues, $D_f(\vec{s}_0)$ cannot 
have purely imaginary eigenvalues, and thus $\vec{s}_0$ is hyperbolic.
\end{proof}

From the previous theorem we know that all critical points $\rho_0$
of $V$ are in fact hyperbolic fixed points of the dynamical system.
It is easy to see that among the $n!$ fixed points, $\rho_0=\rho_d$,
which corresponds to $V(\rho_0)=0$, must be a sink, and the point
corresponding to $V(\rho_0)=V_{\rm max}$ must be a source.  Any other 
fixed point $\rho_0$ must be a saddle, with eigenvalues having both 
negative and positive real parts, for otherwise $\rho_0$ would be a
sink or source, and thus a local minimum or maximum of $V$, which
would contradict Theorem \ref{thm:crit:2}.  Each of these saddle 
points has a stable manifold of dimension $<n^2-n$, on which solutions 
$\rho(t)$ will converge to the saddle point, but since the dimension
is less than the dimension of the state manifold, these solutions 
only constitute a measure-zero set.  Hence, for almost any flow 
$\rho(t)$ outside $E$ will converge to $\rho_d$.  In this sense, the 
Lyapunov control is still effective.  

\begin{remark}
Since the critical points of the dynamical system (\ref{eqn:auto}) for
a generic stationary state $\rho_d$ are hyperbolic and they are also 
hyperbolic critical points of the function $V(\rho)=V(\rho,\rho_d)$, 
the dimension of the stable manifold at a critical point must be the 
same as the index number of the critical point of the function $V$.
\end{remark}

\subsubsection{Generic non-stationary target state}

For non-stationary states characterizating the invariant set is more
complicated as $E$ may now contain points with nonzero diagonal 
commutators.

\begin{example}
Let $\rho_2=\rho_d(0)$ and consider
\begin{eqnarray*}
\rho_1= \begin{pmatrix}
\frac{1}{12} & -\frac{1}{12} & -\frac{1}{12}  \\
-\frac{1}{12} & \frac{11}{24} & \frac{1}{8} \\
-\frac{1}{12} & \frac{1}{8}   & \frac{11}{24}
\end{pmatrix}, \quad
\rho_2= \begin{pmatrix}
\frac{1}{3} & -\frac{i}{12} & \frac{i}{12} \\
\frac{i}{12} & \frac{1}{3}  & -\frac{i}{4} \\
-\frac{i}{12} & \frac{i}{4} & \frac{1}{3}
\end{pmatrix}.
\end{eqnarray*}
$\rho_1$ and $\rho_2$ are isospectral and 
$[\rho_1,\rho_2]=\frac{11i}{144}\diag(0,1,-1)$ and thus
$(\rho_1,\rho_2)\in E$.
\end{example}

Simulations suggest that Lyapunov control is ineffective, i.e., fails 
to steer $\rho(t)$ to $\rho_d(t)$ or even the orbit of $\rho_d(t)$ in 
such cases.  However, it is difficult to give a rigorous proof of this
observation, as we lack a constructive method to ascertain asymptotic
stability near a non-stationary solution.  In the special case where 
$\rho_d(t)$ is periodic there are tools such as Poincar\'e maps but it 
is difficult to write down an explicit form of the Poincar\'e map for 
general periodic orbits~\cite{Perko}.  Moreover, as observed earlier,
for $n>2$ the orbits of non-stationary target states $\rho_d(t)$ under
$H_0$ are periodic only in some exceptional cases.  Fortunately though,
we shall see that $E=\{[\rho_1,\rho_2]=0\}$ still holds for a very 
large class of generic target states $\rho_d(t)$, and in these cases
Lyapunov control tends to be effective.

%In general, we can say only that for points in the invariant set the
%absolute values of each entry of $\rho_2$ must be the same as that of
%the corresponding entry of $\rho_d(0)$, since, if the $(j,k)$th entry 
%of $\rho_d(0)$ is $q_{jk}$ then the $(j,k)$th entry of
%$\rho_d(t)=e^{-iH_0t}\rho_d(0)e^{iH_0t}$ is
%$q_{jk}e^{-i(\omega_j-\omega_k)t}$.

Noting $[\rho_1,\rho_2]=-\Ad_{\rho_2}(\rho_1)$, where $\Ad_{\rho_2}$ 
is a linear map from the Hermitian or anti-Hermitian matrices into 
$\su(n)$, let $A(\vec{s}_2)$ be the real $(n^2-1)\times (n^2-1)$ 
matrix corresponding to the Stokes representation of $\Ad_{\rho_2}$. 
Recall $\su(n)=\T\oplus\C$ and $\RR^{n^2-1}=S_\T\oplus S_\C$, where 
$S_\C$ and $S_\T$ are the real subspaces corresponding to the Cartan 
and non-Cartan subspaces, $\C$ and $\T$, respectively.  Let 
$\tilde{A}(\vec{s}_2)$ be the first $n^2-n$ rows of $A(\vec{s}_2)$ 
(whose image is $S_\T$).

\begin{lemma}
For a generic $\rho_d(t)$ the invariant set $E$ contains points with
nonzero commutator if and only if $\rank\tilde{A}(\vec{s}_d(0))<n^2-n$.
\end{lemma}

\begin{proof}
It suffices to show that if $\rank \tilde{A}(\vec{s}_d)=n^2-n$, 
then for any $\rho$ such that $[\rho,\rho_d(0)]$ diagonal, we have
$[\rho,\rho_d(0)]=0$.  If this is true then for any 
$(\rho_1,\rho_2)\in E$ with diagonal commutator, there exists some
$t_0$ such that $\rho_2=e^{iH_0t_0}\rho_d(0)e^{-iH_0t_0}$ and since
$[\rho_1,\rho_2]$ is diagonal, 
$[e^{-iH_0t_0}\rho_1 e^{iH_0t_0},\rho_d(0)]$ is also diagonal, 
hence equal to zero and $[\rho_1,\rho_2]=0$. 

Let $\rho_2=\rho_d(0)$.  First we show that the kernel of $A(\vec{s}_2)$ 
has dimension $n-1$ and thus $\rank A(\vec{s}_2) \le n^2-n$.  In this 
case $\rank \tilde{A}(\vec{s}_d) = n^2-n = \rank A(\vec{s}_2)$ implies 
that the remaining $n-1$ rows of $A(\vec{s}_2)$ are linear combinations 
of the rows of $\tilde{A}(\vec{s}_2)$ and therefore 
$\tilde{A}(\vec{s}_2)\vec{s}_1=\vec{0}$ implies 
$A(\vec{s}_2)\vec{s}_1=\vec{0}$, or $[\rho_1,\rho_2]=0$.

In order to show that the kernel of $A(\vec{s}_2)$ has dimension
$n-1$, we recall that if $\rho_2=U \diag(w_1,\ldots,w_n) U^\dagger$
for some $U\in\SU(n)$ then $[\rho_1,\rho_2]=0$ for all 
$\rho_1=U\diag(w_{\tau(1)},\ldots,w_{\tau(n)}) U^\dagger$, where 
$\tau$ is a permutation of $\{1,\ldots,n\}$.  If the $w_k\ge 0$ are 
distinct then these $\rho_1$'s span at least a subspace of dimension 
$n$ since the determinant of the circulant matrix
\begin{equation*}
  C = \begin{pmatrix}
       w_1 & w_2 & \ldots & w_{n-1} & w_n \\
       w_2 & w_3 & \ldots & w_n & w_1 \\
       \vdots & \vdots & \ddots & \vdots & \vdots \\
       w_{n-1} & w_n & \ldots & w_{n-3} & w_{n-2} \\
       w_n & w_1 & \ldots & w_{n-2} & w_{n-1}
      \end{pmatrix}
\end{equation*}
is non-zero, and hence its columns are linearly independent and span
the $n$-dimensional subspace of diagonal matrices.  If the $w_k$ are
distinct then the kernel cannot have dimension greater than $n-1$
since the $\rho_1$ can only span a subspace isomorphic to the set of
diagonal matrices. Thus, the kernel of $A(\vec{s}_2)$ has dimension
$n-1$.  (The dimension is reduced by one since we drop the projection 
of $\rho$ onto the identity in the Stokes representation.)  Similarly, 
we can prove if $\rank \tilde{A}(\vec{s}_d(0))<n^2-n$, then $E$ 
contains points with nonzero commutator.
\end{proof}

This lemma provides a necessary and sufficient condition on $\rho_d(0)$ 
to ensure that $[\rho_1,\rho_2]$ diagonal implies $[\rho_1,\rho_2]=0$.
Assuming the first $n^2-n$ rows correspond to $S_\T$, let $\tilde{A}_1$ 
be the submatrix generated from the first $n^2-n$ rows and last $n^2-n$ 
columns of $\tilde{A}(\vec{s}_d(0))$.  If $\det(\tilde{A}_1)\ne 0$ then 
$\rank \tilde{A}(\vec{s}_d(0))=n^2-n$, hence $E=\{[\rho_1,\rho_2]=0\}$.
We can easily verify that if the diagonal elements of $\rho_d(0)$ are not 
equal then $\det(\tilde{A}_1)$ is a non-trivial polynomial, i.e., 
$\det(\tilde{A}_1)$ can only have a finite set of zeros.  Hence we have:

\begin{theorem}
The invariant set $E$ for a generic $\rho_d(t)$ contains points with 
nonzero commutator only if either $\rho_d$ has some equal diagonal 
elements or $\det(\tilde{A}_1)=0$.  Therefore, the set of $\rho_d(0)$ 
such that $E$ contains points with nonzero commutator has measure zero 
with respect to the state space $\M$.
\end{theorem}

Hence, if we choose a generic target state $\rho_d(0)$ randomly, with 
probability one, it will be such that $E=\{[\rho_1,\rho_2]=0\}$.
Simulations suggests Lyapuonv control is generally effective in this
case, and we shall now prove this.  Let $\tau_k$ for $k=1,\ldots,n!$ 
denote all the permutations of the numbers $\{1,\ldots,n\}$ with 
$\tau_1$ being the identity permutation and $\tau_{n!}$ being the 
inversion.  For any given density operator 
\begin{equation}
  \rho(t)=\sum_{m=1}^n w_m \ket{m}\bra{m}, 
\end{equation}
define the `permutation'
\begin{equation}
 \rho^{(k)}(t)=\sum_{m=1}^n w_{\tau_k(m)} \ket{m}\bra{m}.
\end{equation}

\begin{theorem}
\label{thm:generic:conv1}
If $\rho_d(t)$ is a generic state with invariant set
$E=\{[\rho_1,\rho_2]=0\}$ then any solution $\rho(t)$ converges to 
$\rho_d^{(k)}(t)$ for some $k\in \{1,\ldots,n!\}$, and all
solutions except $\rho_d^{(1)}(t)=\rho_d(t)$, which is stable, are
unstable.
\end{theorem}

\begin{proof}
For any solution $(\rho(t),\rho_d(t))$ there exists a subsequence 
$\{t_m\}$ such that $(\rho(t_m),\rho_d(t_m))\to (\rho_1,\rho_2)\in E$. 
If $E$ only contains pairs $(\rho_1,\rho_2)$ that commute then we can 
choose an orthonormal basis such that both $\rho_1$ and $\rho_2$ are 
diagonal, and since $\rho_1$ and $\rho_2$ have the same spectrum, the 
diagonal elements of $\rho_1$ must be a permutation of those of $\rho_2$, 
i.e., $\rho_1=\rho_2^{(k)}$ for some $k$.  Thus we have
$\rho(t_m)\to \rho_1=\rho_2^{(k)}$, $\rho_d(t_m) \to \rho_2$ and 
therefore
\begin{equation}
\label{eqn:1}
 \rho(t_m)\to \rho_d^{(k)}(t_m).
\end{equation}
If $(\bar\rho_1,\bar\rho_2)\in E$ is a different positive limiting
point of $(\rho(t),\rho_d(t))$, we can similarly find a subsequence
$\{t_{m'}\}$ such that $\rho(t_{m'})\to \rho_d^{(k')}(t_{m'})$, for
some $k'$. Since $V(\rho(t),\rho_d(t))$ is non-increasing along the
trajectory, we must have $k=k'$. Therefore, the result~(\ref{eqn:1})
holds for any subsequence $\{t_m\}$. 

To see that all solutions except those with $\rho(t)\to \rho_d(t)$ are 
unstable, we consider the dynamics in the interaction picture.  Let
\begin{eqnarray*}
\bar \rho_d(t) &=& e^{iH_0t}\rho_d(t)e^{-iH_0t}=\rho_d(0)\\
\bar \rho(t)   &=& e^{iH_0t}\rho_d(t)e^{-iH_0t}.
\end{eqnarray*}
We have $\dot {\bar \rho}_d(t)=0$ and the dynamical system becomes:
\begin{subequations}
\label{eqn:inter}
\begin{align}
\dot {\bar\rho}(t)  &=\bar f(t)[-i\bar H_1(t), \bar\rho(t)]\\
          \bar f(t) &=\Tr([-i\bar H_1(t), \bar\rho(t)]\bar\rho_d)
\end{align}
\end{subequations}
where $\bar H_1(t)=e^{iH_0t}H_1(t)e^{-iH_0t}$ and $\bar f(t)=f(t)$.
Thus, the original autonomous dynamical system $(\rho(t),\rho_d(t))$, 
where $\rho_d$ is not stationary, has transformed into a non-autonomous 
system, where $\bar\rho_d$ is a fixed point.  According to Theorems 
\ref{thm:generic:crit} and \ref{thm:generic:hyperbolic}, for a given 
$\bar\rho_d$ there are $n!$ hyperbolic critical points of the function 
$V(\bar\rho)=V(\bar\rho,\bar\rho_d)$, denoted by $\bar\rho_d^{(k)}$, 
$k=1,\ldots,n!$, with $\bar\rho_d^{(1)}=\bar\rho_d$ and $\bar\rho_d^{(n!)}$ 
corresponding to the minimum and maximum, respectively.  They are also 
the fixed points of the dynamical system~(\ref{eqn:inter}).  

If $E=\{[\rho_1,\rho_2]=0\}$ then any solution $\bar\rho(t)$ must 
converge one of the critical points $\bar\rho_d^{(k)}$.  Since the 
fixed points of the dynamical system~(\ref{eqn:inter}) coincide with 
the $n!$ hyperbolic critical points of $V(\bar\rho)$ for a given 
$\bar\rho_d$ and $V$ is non-increasing along any solution, it is 
easy to see that $\bar\rho_d$ and $\bar\rho_d^{(n!)}$ correspond to 
a stable and unstable point, respectively.  For any other fixed point 
$\bar\rho_d^{(k)}$, if it is stable, it must be asymptotically stable 
since all solutions must converge to one of these fixed points.  
However, by the continuity of the function $V$, this would imply that 
$\bar\rho_d^{(k)}$ is a local minimum, which is a contradiction to 
the fact that it is a hyperbolic saddle of $V(\bar\rho)$.  Therefore, 
all the 'intermediate' fixed points are unstable for the 
system~(\ref{eqn:inter}), and therefore $\rho_d^{(k)}(t)$, 
$k=2,\ldots,n!-1$, must be unstable.
\end{proof}

Numerical simulations for non-stationary target states $\rho_d(t)$ 
such that $E=\{[\rho_1,\rho_2]=0\}$ suggest that almost all solutions 
$\bar\rho(t)$ converge to $\bar\rho_d$, which is consistent with the 
theorem.  However, unlike for the stationary case we cannot conclude 
that the solutions converging to the saddles between the maximum and 
minimum constitute a measure-zero set.  We can still show, though, 
that in principle there exist solutions $\bar\rho(t)$ starting very 
close to $\bar\rho_d^{(n!)}(t)$ that converge to the target state 
$\bar\rho_d$.  So the region of asymptotic stability of $\rho_d(t)$ 
is at least not confined to a local neighborhood of it.

\begin{theorem}
In the interaction picture~(\ref{eqn:inter}) for any saddle point
$\bar\rho_d^{(k_0)}$ with $1<k_0<n!$ not all solutions $\bar\rho(t)$ 
with $V(\bar\rho(0))>V(\bar\rho_d^{(k_0)})$ can converge to it.
\end{theorem}

\begin{proof}
From the topological structure near a hyperbolic saddle point we know 
that the pre-image of $V=V(\bar\rho_d^{(k_0)})$ contains not only 
$\bar\rho_d^{(k_0)}$.  Therefore, we can choose a $\bar\rho_0$ such 
that $V(\bar\rho_0)=V(\bar\rho_d^{(k_0)})$ and 
$\bar\rho_0\neq\bar\rho_d^{(k_0)}$.  Since the solution $\bar\rho(t)$ 
with $\bar\rho(0)=\bar\rho_0$ cannot be a stationary, there exists a 
time $t_1<0$ such that $V(\bar\rho(t_1))>V(\bar\rho_d^{(k_0)})$.
\end{proof}

\subsection{Other stationary target states}

We have shown that for `ideal' systems, Lyapunov control is mostly
effective for both pseudo-pure and generic states, which covers the
largest and most important classes of states.  Finally, we show that 
if $\rho_d$ is stationary but has degenerate eigenvalues then there
may be large critical manifolds but we can still derive a result 
similar to the asymptotic stability of $\rho_d$ in the discussion of 
generic stationary states $\rho_d$. 

\begin{theorem}
\label{thm:degenerate:1}
If $\rho_d$ is a stationary state with degenerate eigenvalues then 
$\rho=\rho_d$ is a hyperbolic critical point of the function 
$V(\rho)=V(\rho,\rho_d)$ and it is isolated from the other critical 
points.
\end{theorem}

\begin{proof}
Choose a basis such that $\rho_d$ is diagonal,
\begin{equation}
\label{eqn:degenerate}
\rho_d=\diag(a_1,\ldots,a_1,a_2,\ldots,a_2,\ldots,a_k,\ldots,a_k),
\end{equation}
and let $n_1,n_2,\ldots,n_k$, denote the multiplicities of the distinct
eigenvalues, where $\sum_{\ell=1}^k n_\ell=n$.  Using the same notation 
as in Theorem~\ref{thm:crit:2}, $\rho=\rho_d$ achieves the maximal value 
of $J=\Tr(\rho\rho_d)=\Tr(U\rho_dU^\dagger\rho_d)$.  To show that it is
a hyperbolic maximum of $J$ (hence minimum of $V$) we need to find $n'$
independent directions along each of which $J$ is a local maximum, where 
$n'$ is the dimension of the manifold $\M$, in our case 
$n'=n^2-\sum_{\ell=1}^k n_\ell^2$.  As in the proof of 
Theorem~\ref{thm:crit:2}, we note that for the curves with
$\xs=\lambda_{k\ell}t$ and $\xs=\bar\lambda_{k\ell}t$, the conjugate
action of $\lambda_{k\ell}$ or $\bar\lambda_{k\ell}$ on the critical
point $\rho=\rho_d$ swaps the $k$-th and $\ell$-th diagonal
elements. Hence the number of swaps that decrease the value of $J$ is
\begin{align*}
 & 2(n_1 \sum_{\ell=2}^k n_\ell 
    +n_2 \sum_{\ell=3}^k n_\ell
    +\cdots+n_{k-1} n_k \\
 &= n^2-\sum_{\ell=1}^k n_\ell^2=n'.
\end{align*}
Therefore, $\rho=\rho_d$ is a hyperbolic point of $J$, hence of $V$.
Since the critical values of $V$ as shown in Eq.~(\ref{eqn:V}), are
isolated and $\rho_d$ is the unique minimal value, it must also be
isolated from the other critical points, which completes the proof.
\end{proof}

Furthermore, we can show that $\rho_d$ is also a hyperbolic fixed point 
for the dynamical system~(\ref{eqn:auto1}):

\begin{theorem}
\label{thm:degenerate:2}
If $\rho_d$ is a stationary state with degenerate eigenvalues then 
$\rho=\rho_d$ is a hyperbolic sink of the dynamical 
system~(\ref{eqn:auto1}).
\end{theorem}

\begin{proof}
As in Theorem~\ref{thm:generic:hyperbolic}, we need to analyze the 
eigenvalues of linearization matrix $D_f(\vec{s}_d)$.  In order to 
show $\vec{s}_d$ is hyperbolic, it suffices to show that there are 
$n_{\M}$ eigenvalues with nonzero real parts, corresponding to 
$n_{\M}$ eigenvectors in the tangent space of $\M$ at $\vec{s}_d$, 
denoted as $T_{\M}(\vec{s}_d)$.  Let $\vec{v}$ be a column vector 
consisting of $n(n-1)/2$ elementary parts:
\begin{equation}
 \vec{v}^{(k,\ell)} =
\frac{\Delta_{k\ell}}{\sqrt{2}}
 \begin{pmatrix}
 \Im(b_{k\ell}) \\ \Re(b_{k\ell})
\end{pmatrix},
\end{equation}
and let $B=B_0-\vec{v}\vec{v}^T$ be the restriction of $D_f(\vec{s}_d)$
to the subspace $S_\T$ as before.  Following a similar argument as in 
Lemma~\ref{lemma:generic:3} it is easy to see that for $(k,\ell)$ such 
that $\Delta_{k\ell}=0$, the eigenelement $(\omega_{k\ell},\vec{e}_{k\ell})$ 
of $B_0$ is also an eigenelement of $B$ as $\vec{v}^T\vec{e}_{k\ell}=0$,
and that $\vec{e}_{k\ell}$ corresponds to a direction orthogonal to the 
tangent space $T_{\M}(\vec{s}_d)$.  The number of such $(k,\ell)$ is
\[
  \bar{N} = 2 \sum_{\ell=1}^k \binom{n_\ell}{2}.
\] 
By same arguments as in the proof of Theorem~\ref{thm:generic:hyperbolic}, 
it therefore is easy to show that the remaining eigenvalues of $B$ with 
eigenvectors corresponding to the directions in $T_{\M}(\vec{s}_d)$ must 
have non-zero real parts.  A simple counting argument shows that the 
number of these eigenvalues is $2\binom{n}{2}-\bar n=\dim(\M)$ and thus
$\rho_d$ is a hyperbolic point.  Since $\rho_d$ achieves the minimum of 
$V$, these eigenvalues must have negative real parts, i.e., $\rho_d$ must 
be a sink.
\end{proof}

Hence, any solution $\rho(t)$ near $\rho_d$ will converge to $\rho_d$ 
for $t\to +\infty$, which establishes local asymptotic stability of 
$\rho_d$.  The next question is whether this asymptotic convergence 
holds for a larger domain, as in the case of stationary non-degenerate 
$\rho_d$.  In order to answer this, we need to investigate the LaSalle 
invariant set.  For a stationary $\rho_d$ with degenerate eigenvalues
as in Eq.~(\ref{eqn:degenerate}) there are $p=\frac{n!}{n_1!\cdots n_k!}$ 
distinct diagonal $\rho_0$ satisfying $[\rho_0,\rho_d]=0$.  First of 
all, we have the following lemma:

\begin{lemma}
\label{lemma:degenerate}
Any stationary point $\rho_0$ other than $\rho_d$ must correspond to
a maximum along some direction.
\end{lemma}

\begin{proof}
Since $\rho_0\ne \rho_d$, we have $V(\rho_0)>V(\rho_d)$, and
analogous to the proof in Theorem \ref{thm:crit:2}, there exists 
some swap $\lambda_{k\ell}$ such that $\rho_0$ corresponds to a 
maximum along that direction.
\end{proof}

Furthermore, we can prove that the LaSalle invariant set consists of
centre manifolds with the diagonal stationary states $\rho_0$ as
centres.  This can be easily illustrated with the following example:

\begin{example}
For a three-level system with 
$\rho_d=\diag(\frac{1}{4},\frac{1}{4},\frac{1}{2})$, the dimension
of the state manifold $\M$ is $\dim{M}=3^2-2^2-1=4$ and the LaSalle
invariant set $E$ contains all points $\rho_0$ of the form
\begin{align*}
 \rho_0= 
  \begin{pmatrix}
  a_{11}   & a_{12} & 0 \\
  a_{12}^* & a_{22} & 0 \\
  0        & 0      & a_{33},
  \end{pmatrix}
\end{align*}
with eigenvalues $\{\frac{1}{4},\frac{1}{4},\frac{1}{2}\}$.  If
$a_{33}=\frac{1}{2}$ then we have $\rho_0=\rho_d$, which is an isolated
hyperbolic sink.  All other $\rho_0$ in $E$ satisfy $a_{33}=\frac{1}{4}$
and form a manifold $\M_0$, which contains two stationary states
\(\rho_1=\diag\big(\frac{1}{4},\frac{1}{2},\frac{1}{4}\big)\) and
\(\rho_2=\diag\big(\frac{1}{2},\frac{1}{4},\frac{1}{4}\big)\) that 
commute with $\rho_d$.

Analogous to the proof of Theorem~\ref{thm:degenerate:2}, we can analyze
the linearization of the dynamical system near one of the critical point
$\rho_\ell$, $\ell=1,2$. It is easy to see that the two tangent vectors
of the centre manifold at $\rho_1$ (corresponding to two purely imaginary
eigenvalues) are also the tangent vectors of the invariant manifold $\M_0$. 
Therefore, $\M_0$ is the centre manifold.  The other two eigenvalues must 
have positive real parts, since $\rho_1$ is the maximal point of $V$. 
This analysis is also true for $\rho_2$.  Hence, except for the target 
state $\rho_d$, the points in the LaSalle invariant set form a centre 
manifold with the stationary points $\rho_\ell$, $\ell=1,2$ as centres.
\end{example}

In general, we can analyze the linearization near any of the
$p=\frac{n!}{n_1!\cdots n_k!}$ stationary points of the dynamical
system.  For a stationary point $\rho_0$ other than $\rho_d$, analysis 
the eigenvalues of the linearized system, analogous to the previous 
example, shows that the purely imaginary eigenvalues correspond to the 
centre manifolds generated by the LaSalle invariant set, where $\rho_0$ 
is a centre on the centre manifold.  Other eigenvalues can be similarly 
proved to have either positive or negative real parts.  Moreover, the 
spectrum must contain eigenvalues with positive real parts; otherwise, 
$\rho_0$ would be dynamically stable, corresponding to a local minimum 
of $V$, which contradicts Lemma~\ref{lemma:degenerate}.  Hence, near
$\rho_0$, except for the solutions on the stable manifold of $\rho_0$, 
all the other solutions will move away.  Therefore, globally, provided
we start outside the LaSalle invariant set, most solutions will converge 
to $\rho_d$, similar to the results for generic stationary $\rho_d$.

\section{Convergence of Lyapunov Control for realistic systems}
\label{sec:conv_real}

In the previous section we studied the invariant set and convergence
behavior of Lyapunov control for systems that satisfy very strong
requirements, namely complete regularity of $H_0$ and complete
connectedness of the transition graph associated with $H_1$.  We shall
now consider how the invariant set and convergence properties change
when the system requirements are relaxed.  Without loss of generality,
we present the analysis for a qutrit system noting that the
generalization to $n$-level systems is straightforward.

\subsection{$H_1$ not fully connected}

Suppose $H_0$ is strongly regular but $H_1$ does not have couplings
between every two energy levels, i.e., the field does not drive every
possible transition, as is typically the case in practice.  For example,
for many model systems such as the Morse oscillator only transitions
between adjacent energy levels are permitted and we have for $n=3$:
\begin{equation*}
H_0 =\begin{pmatrix}
a_1 & 0 & 0  \\
0 & a_2 & 0  \\
0 & 0 & a_3
\end{pmatrix}, \qquad
H_1=\begin{pmatrix}
0 & b_1 & 0  \\
b_1^* & 0 & b_2  \\
0 & b_2^* & 0
\end{pmatrix}
\end{equation*}
where we may assume $a_1<a_2<a_3$, for instance.

According to the characterization of the invariant set $E$ derived
in Section~\ref{sec:LaSalle}, a necessary condition for
$(\rho_1,\rho_2)$ to be in the invariant set $E$ is that
$[\rho_1,\rho_2]$ is orthogonal to the subspace spanned by the
sequence $\B=\Span\{B_m\}_{m=0}^\infty$ with
$B_m=\Ad_{-iH_0}^{(m)}(-iH_1)$.  Comparison with (\ref{eq:Bm}) shows
that if the coefficient $b_{k\ell}=0$ then none of the generators
$B_m$ have support in the root space $\T_{k\ell}$ of the Lie
algebra, and it is easy to see that the subspace of $\su(n)$
generated by $\B$ is the direct sum of all root spaces $\T_{k\ell}$
with $b_{k\ell}\neq 0$.

Thus, in our example, a necessary condition for $(\rho_1,\rho_2)$ to
be in the invariant set $E$ is $[\rho_1,\rho_2] \in \T_{13}\oplus\C$,
which shows that $[\rho_1,\rho_2]$ must be of the form
\begin{equation}
 \label{eq:rho13}
 [\rho_1,\rho_2]= \begin{pmatrix}
  \alpha_{11} & 0 & \alpha_{13}  \\
  0 & \alpha_{22} & 0  \\
  \alpha_{13}^* & 0 & \alpha_{33}
\end{pmatrix}.
\end{equation}
Furthermore, if $(\rho_1,\rho_2)$ is of type~(\ref{eq:rho13}) then
\begin{equation}
 U_0(t) [\rho_1,\rho_2] U_0(t)^\dagger
 = \begin{pmatrix}
  \alpha_{11} & 0 & e^{i\omega_{13}t}\alpha_{13}  \\
  0 & \alpha_{22} & 0  \\
 e^{-i\omega_{13}t}\alpha_{13}^* & 0 & \alpha_{33}
\end{pmatrix}
\end{equation}
with $U_0=e^{-iH_0t}$ and $\omega_{k\ell}=a_\ell-a_k$, also has the
form.  Therefore, $[\rho_1,\rho_2]\in \C\oplus\T_{13}$ is a necessary
and sufficient condition for the invariant set $E$.

If $\rho_d$ is diagonal with non-degenerate eigenvalues then $E$
consists of all $(\rho_1,\rho_2)$ with $\rho_2=\rho_d$ and $\rho_1$ 
of the form
\begin{equation}
\rho_1= \begin{pmatrix}
       \beta_{11} & 0 & \beta_{13}  \\
       0 & \beta_{22} & 0  \\
       \beta_{13}^* & 0 & \beta_{33}
\end{pmatrix}.
\end{equation}
Thus, the invariant set $E$ contains a finite number of isolated
fixed points corresponding to $\beta_{13}=0$, which coincide with
the critical points of $V(\rho,\rho_d)$ as a function on the
homogeneous space $\M \times \M$ with 
$\M\simeq \UU(3)/\{\exp(\sigma):\sigma\in\C\}$
for fixed $\rho_d$, as well as an infinite number of trajectories
with $\beta_{13}\ne 0$.

We check the stability of linearized system near these fixed points,
concentrating on the local behavior near $s_d$.  Working with a real
representation of the linearized system~(\ref{eqn:linear}) and using
the same notation as before, we can still show that $D_f(\vec{s}_d)$
has $n^2-n$ nonzero eigenvalues, $n$ equal to three in our case.
Since $-iH_1$ has no support in the root space $\T_{13}$, the
$\lambda_{13}$ and $\bar\lambda_{13}$ components of $A_1\vec{s}_d$,
(which correspond to $[-iH_1,\rho_d]$) vanish, and $D_f(\vec{s}_d)$
has a pair of purely imaginary eigenvalues whose eigenspaces span
the root space $\T_{13}$ and four eigenvalues with non-zero real
parts, which must be negative since $\vec{s}_d$ is locally stable
from the Lyapunov construction.  However, the existence of two
purely imaginary eigenvalues means that the target state is no
longer a hyperbolic fixed point but there is centre manifold of
dimension two.  From the centre manifold theory, the qualitative
behavior near the fixed point is determined by the qualitative
behavior of the flows on the centre manifold~\cite{Carr}. Therefore,
the next step is to determine the centre manifold.  For dimensions
greater than two, this is generally a hard problem if we do not know
the solution of the system.  However, since we know the tangent
space of the centre manifold, if we can find an invariant manifold
that has this tangent space at $\vec{s}_d$, then it is a centre
manifold.

In our case solutions in the invariant set $E$ form a manifold that
is diffeomorphic to the Bloch sphere for a qubit system, with the
natural mapping embedding
\begin{equation}
\rho=\begin{pmatrix}
\beta_{11} & 0 & \beta_{13}  \\
0 & \beta_{22} & 0  \\
\beta_{13}^* & 0 & \beta_{33}
\end{pmatrix}
\to \rho'=
\frac{1}{\beta_{11}+\beta_{33}}
\begin{pmatrix}
\beta_{11} & \beta_{13}\\
\beta_{13}^* &  \beta_{33}
\end{pmatrix},
\end{equation}
which maps the state $\rho_d$ (or $\vec{s}_d$) of the qutrit to the
point $\vec{s}_d'$ on the Bloch sphere corresponding to
$\rho_d'=\diag(w_1,w_3)/(w_1+w_3)$, and the two tangent vectors of the
centre manifold at $\rho_d$ to the two tangent vectors of the Bloch
sphere at $\vec{s}_d'$.  Thus this manifold is the required centre
manifold at $\rho_d$ (or $\vec{s}_d$).  On the centre manifold $\rho_d$
is a centre with the nearby solutions cycling around it.  The 
Hartman-Grobman theorem in centre manifold theory proved by
Carr~\cite{Carr} shows that all solutions outside $E$ converge
exponentially to solutions on the centre manifold belonging to
$\vec{s}_d$, while the solutions actually converging to $\vec{s}_d$ only
constitute a set of measure zero.  Therefore, when $H_1$ is not fully
connected, the trajectories $\rho(t)$ for most initial states $\rho(0)$
will not converge to the target state $\rho_d$ (or another critical
point of $V$) but to other trajectories $\rho_1(t)\subset E$, which are
not in orbit of $\rho_d$ either.

\subsection{$H_0$ not strongly regular}

Next let us consider systems with $H_1$ fully connected but $H_0$
not strongly regular, such as
\begin{equation}
H_0= \begin{pmatrix}
0 & 0 & 0  \\
0 & \omega & 0  \\
0 & 0 & 2\omega
\end{pmatrix}, \quad
H_1=\begin{pmatrix}
0 & 1 & 1  \\
1 & 0 & 1  \\
1 & 1 & 0
\end{pmatrix}
\end{equation}
i.e., $\omega_{12}=\omega_{23}=\omega$.  In order to determine the
subspace spanned by $\B=\{B_m\}_{m=0}^\infty$ [See~(\ref{eq:Bm})],
we note that the characteristic Vandermonde
matrix~(\ref{eq:vandermonde}) of the system
\begin{equation*}
V =
\begin{pmatrix}
1 & \omega^2 & \omega^4  \\
1 & (2\omega)^2 & (2\omega)^4 \\
1 & \omega^2 & \omega^4
\end{pmatrix}
\end{equation*}
has rank two as only the first two rows are linearly independent. We
find that in this case the invariant set $E$ is characterized by
$[\rho,\rho_d]\in\C\oplus\Span\{\mu,\bar{\mu}\}$ with
$\mu=\lambda_{12}-\lambda_{23}$,
$\bar{\mu}=\bar\lambda_{12}-\bar{\lambda}_{23}$.
\begin{equation*}
[\rho,\rho_d]= \begin{pmatrix}
0                      & -\beta_{12}\Delta_{12}  & -\beta_{13}\Delta_{23}  \\
\beta_{12}^*\Delta_{12}& 0                       & -\beta_{23}\Delta_{23}  \\
\beta_{13}^*\Delta_{23}& \beta_{23}^*\Delta_{23}& 0
\end{pmatrix}
\end{equation*}
where $\Delta_{k\ell}=w_k-w_\ell$.
Thus $[\rho,\rho_d]\in \C\oplus\Span\{\mu,\bar\mu\}$ implies
$\beta_{13}=0$ and $\beta_{12}\Delta_{12}=-\beta_{23}\Delta_{23}$.
So all $\rho\in E$ form a two-dimensional manifold with coordinates
determined by the $\lambda_{12}$ and $\bar\lambda_{12}$ components
of $\rho$.

As we are interested in the local dynamics near the target state, we
again study the linearization at the fixed point $\rho_d$ for the
case of a generic stationary state, i.e., $\rho_d$ diagonal with
non-degenerate eigenvalues.  Using the same notation as before, the
matrix $B_0$, i.e., the restriction of $\Ad_{-iH_0}$ to the subspace
$S_\T$, has six non-zero eigenvalues $\{\pm i\omega,\pm 2i\omega\}$,
where $\pm i\omega$ occurs with multiplicity two, and since
$\det(B)=\det(B_0)$, we know that $B$ also has six non-zero
eigenvalues.  However, two of these are purely imaginary, namely
$\pm i\omega$, as it can easily be checked that $\det(B\pm i\omega
I)=0$, and the corresponding vectors are
\begin{equation*}
 \vec{e}_{\pm i\omega}  = (-\Delta,\mp i\Delta,0,0,1,-i)^T
\end{equation*}
where $\Delta=\Delta_{23}/\Delta_{12}$.  Moreover, we know that all
other eigenvalues of $B$ must have negative (non-zero) real parts.
Analogous to the last subsection, we can show that the invariant set
$E$ forms a centre manifold near $\vec{s}_d$ with $\vec{s}_d$ as a
centre.  Thus by the Hartman-Grobman theorem of the centre manifold
theory, we can again infer that most of the solutions near
$\vec{s}_d$ will not converge to $\vec{s}_d$.

\section{Conclusions and Further discussions}

We have presented a detailed analysis of the Lyapunov method for the
problem of steering a quantum system towards a stationary target state,
or tracking the trajectory of a non-stationary target state under free
evolution, for finite-dimensional quantum systems governed by a bilinear
control Hamiltonian.  Although our results are partially consistent with
previously published work in the area, our analysis suggests a more 
complicated picture than previously described.

First, to allow proper application of the LaSalle invariance principle
we transform the original control problem into an autonomous dynamical
systems defined on an extended state space.  Characterization of the
LaSalle invariant set for this system shows that it always contains the
full set of critical points $F$ of the distance-like Lyapunov function
$V(\rho_1,\rho_2)=\frac{1}{2}\norm{\rho_1-\rho_2}^2$ defined on the
extended state space $\M\times\M$, where $\M$ is the appropriate flag
manifold for the density operators $\rho_1,\rho_2$.  Consistent with
previous work we show that the critical points of $V$ are the only points
in the invariant set for ideal systems, i.e., systems with strongly
regular drift Hamiltonian $H_0$ and fully connected control Hamiltonian
$H_1$, and stationary target states $\rho_d$.  However, we also show
that the invariant set is larger for non-stationary target states or
non-ideal systems, the main difference being that for ideal systems,
there is only a measure-zero set of target states for which the
invariant set $E$ is larger than $F$, while for non-ideal systems the
invariant set is always significantly larger than $F$.  This observation
is important because numerical simulations sugggest that Lyapunov
control design is mostly effective if the invariant set is limited to
the critical points of $V$, but likely to fail otherwise.  Our analysis
for various cases explains why.

For a generic target state (stationary or not) there is always a finite
set of $n!$ critical points of $V$, and it can be shown using stability
analysis that all of these critical points, except the target state, are
unstable.  Specifically, for a stationary generic target state we can
show that all the critical points are hyperbolic critical points of $V$
and hyperbolic critical points of the dynamical system, with the target
state being the only hyperbolic sink.  All the other critical points are
hyperbolic saddles, except for one hyperbolic source corresponding to
the global maximum.  Although this picture is somewhat similar to that
presented in~\cite{altafini2}, our dynamical systems analysis shows the
other critical points, referred to as antipodal points
in~\cite{altafini2}, are unstable, but except for the global maximum,
not repulsive.  In fact, all the hyperbolic saddles have stable
manifolds of positive dimension.  Thus, the set of initial states that
do not converge to the target state, even in this ideal case, is larger
than the (finite) set of antipodal points itself, although for ideal
systems and generic stationary target states, it is a measure-zero set
of the state space.  For stationary states with degenerate eigenvalues
(non-generic states) the set of critical points is much larger, forming
a collection of multiple critical manifolds.  However, for ideal systems
we can show that even in this case the target state is the still the
only hyperbolic sink of the dynamical system and asymptotically stable.
Thus, in general we can still conclude that most states will converge to
the target state, although it is non-trivial to show that the set of
states that converge to points on the other critical manifolds has
measure zero, except for the class of pseudo-pure states.  This class is
special since the set of critial points in this case has only two
components: a single isolated point corresponding to the global minimum
of $V$, which is a hyperbolic sink of the dynamical system, and a
critical manifold homeomorphic to $\CC P^{n-2}$ for $\M=\CC P^{n-1}$, on
which $V$ assumes its global maximum value $V_{\rm max}$.  Thus, although
the points comprising the critical manifold are not repulsive, since $V$
is decreasing as function of $t$, no initial states outside this manifold 
can converge to it.  We note that this argument was employed
in~\cite{altafini2} to argue that the critical points other than the
target state are `repulsive' but our analysis shows that it works only
for the class of pseudo-pure states.

Thus, although our analysis suggest that, e.g., that there are initial
states other than the antipodal points that will not converge to the
target state even for ideal systems, the set of states for which the
Lyapunov control fails is small, except for a measure-zero set of target
states for which the invariant set contains non-critical points.  For
ideal system one could therefore conclude that the Lyapunov method is
overall an effective control strategy.  However, most physical systems
are not ideal, and the Hamiltonians $H_0$ and $H_1$ are unlikely to
satisfy the very stringent conditions of strong regularity and full
connectedness, respectively.  For instance, these assumptions rule out
all systems with nearest-neighbour coupling only, as well as any system
with equally spaced or degenerate energy levels, despite the fact that
most of these systems can be shown to be completely controllable as
bilinear Hamiltonian control systems.  In fact, the requirements for
complete controllability are very low. Any system with strongly regular
drift Hamiltonian $H_0$, whose transition graph is not disconnected, for
instance, is controllable~\cite{CP267p001}, and in many cases even much
weaker requirements suffice~\cite{PRA63n063410, prep1}.  In practice, a 
bilinear Hamiltonian system can generally fail to be controllable only 
if it is decomposable into non-interacting subsystems or has certain 
(Lie group) symmetries, ensuring that, e.g., the dynamics is restricted 
to a subgroup such as the symplectic group~\cite{JPA35p2327}.  

Unfortunately, our analysis shows that the picture changes drastically
for non-ideal systems, with the target state ceasing to be a hyperbolic
sink of the dynamical system and becoming a centre on a centre manifold
contained in a significantly enlarged invariant set $E$.  Using results
from centre manifold theory, we must conclude that most of the solutions
$\rho(t)$ converge to solutions on the centre manifold other than the
target state $\rho_d$.  This result casts serious doubts on the
effectiveness of the Lyapunov method for realistic systems, in fact, it
strongly suggests that Lyapunov control design is an effective method
only for a very small subset of controllable quantum systems.  These
results appear to be in conflict with some recently published results on
Lyapunov control, which suggest that when the Hamiltonian and target
state satisfy a certain algebraic condition then any state $\rho(0)$
that is not an `antipodal' point of $\rho_d(0)$ asymptotically converges
to the orbit of the target state $\rho_d(t)$~\cite{altafini2}, and
claims that the `antipodal' points are repulsive.  Since the notion of
orbit convergence that was used in~\cite{altafini2} is weaker than the
notion of convergence in the sence of trajectory tracking we have used,
one might conjecture this to be the source of the discrepancy, and since
orbit tracking may be quite adequate for many control problems that do
not require precise phase control, for instance, this could mean that
Lyapunov control might still be an effective control strategy for many
quantum control problems.  However, this does not appear to be the case
here.  For instance, the notions of orbit and trajectory tracking are
identical for stationary target states but even for ideal systems and
stationary generic target states that satisfy the conditions
in~\cite{altafini2}, our analysis suggests that the antipodal points,
except one global maximum, are hyperbolic saddle points and hence
unstable but not repulsive.  Furthermore, careful analysis of our
results shows that for ideal systems convergence of $\rho(t)$ to the
orbit of $\rho_d(t)$ implies $\rho(t)\to\rho_d(t)$ except for a
measure-zero set of target states $\rho_d(t)$.\\

\section*{Acknowledgments}

XW is supported by the Cambridge Overseas Trust and an Elizabeth Cherry
Major Scholarship from Hughes Hall, Cambridge.  SGS acknowledges UK
research council funding from an \mbox{EPSRC} Advanced Research
Fellowship and additional support from the \mbox{EPSRC QIP IRC} and 
Hitachi.  She is currently also a Marie Curie Fellow under the European
Union Knowledge Transfer Programme MTDK-CT-2004-509223.  We sincerely
thank Peter Pemberton-Ross, Tsung-Lung Tsai, Christopher Taylor, Jack 
Waldron, Jony Evans, Dan Jane, Yaxiang Yuan, Jonathan Dawes,
Lluis Masanes, Rob Spekkens, Ivan Smith for interesting and fruitful 
discussions.

\end{document}